\documentclass[journal,twoside,web]{ieeecolor}
\usepackage{generic}
\usepackage{cite}
\usepackage{amsmath,amssymb,amsfonts}
\usepackage{algorithmic}
\usepackage{graphicx}
\usepackage{textcomp}
\usepackage{algorithm,algorithmic}
\usepackage{hyperref}
\hypersetup{hidelinks=true}
\newtheorem{theorem}{\textbf{Theorem}}
\newtheorem{definition}{Definition}
\newtheorem{lemma}{Lemma}
\newtheorem{corollary}{Corollary}

\newtheorem{proposition}{Proposition}
\newtheorem{remark}{Remark}
\newtheorem{condition}{Condition}

\def\BibTeX{{\rm B\kern-.05em{\sc i\kern-.025em b}\kern-.08em
    T\kern-.1667em\lower.7ex\hbox{E}\kern-.125emX}}
\begin{document}
\title{Stability Analysis of Distributed Estimators for Large-Scale Interconnected Systems: Time-Varying and Time-Invariant Cases}
\author{Zhongyao Hu, Bo Chen, \IEEEmembership{Senior~Member, IEEE}, Jianzheng Wang, Daniel W. C. Ho, \IEEEmembership{Life Fellow, IEEE}, Wen-An Zhang, \IEEEmembership{Senior~Member, IEEE}, and Li Yu, \IEEEmembership{Senior~Member, IEEE}
%\thanks{Manuscript received XXX; accepted XXX. Date of publication XXX; date of current version XXX. This work was supported in part by the National Natural Science Funds of China under Grant 61973277 and Grant 62073292, in part by the Zhejiang Provincial Natural Science Foundation of China under Grant LR20F030004, in part by the Key Research and Development Program of Zhejiang Province under Grant 2022C03029 and Grant
%2023C01144, in part by the Research Grants Council of the Hong Kong Special Administrative Region, China (CityU 11213023, 11205724). (Corresponding author: Bo Chen.)}% <-this % stops a space
\thanks{Z. Hu, B. Chen, J. Wang, W. Zhang and L. Yu are with Department of Automation, Zhejiang University of Technology, Hangzhou 310023, China. (email: bchen@aliyun.com).}
\thanks{Daniel W. C. Ho is with Department of Mathematics, City University of Hong Kong, Hong Kong, SAR.}
}

\maketitle

\begin{abstract}
This paper studies a distributed estimation problem for time-varying/time-invariant large-scale interconnected systems (LISs). A fully distributed estimator is presented by recursively solving a distributed modified Riccati equation (DMRE) with decoupling variables. By partitioning the LIS based on the transition matrix's block structure, it turns out that the stability of the subsystem is independent of the global LIS if the decoupling variable is selected as the number of out-neighbors. Additionally, it is revealed that any LIS can be equivalently represented by a Markov system. Based on this insight, we show that the stability decoupling above can also be achieved if the decoupling variable equals the number of in-neighbors. Then, the distributed estimator is proved to be stable if the DMRE remains uniformly bounded. When the LIS is considered time-invariant, and by analyzing the spectral radius of a linear operator, it is proved that the DMRE is uniformly bounded if and only if a linear matrix inequality is feasible. Based on the boundedness result, we also show that the distributed estimator converges to a unique steady state for any initial condition. Finally, simulations verify the effectiveness of the proposed methods.
\end{abstract}

\begin{IEEEkeywords}
Distributed estimation, large-scale interconnected systems, stability anslysis.
\end{IEEEkeywords}

\section{Introduction}
\label{sec:introduction}
Large-scale interconnected systems (LISs), which consist of a large number of local subsystems distributed and interconnected in a spatial domain, can be frequently found in fields such as robotics \cite{Feddema1068004} and smart grids \cite{riverso2012hycon2}. To ensure that the fundamental control \cite{FARINA201838} and detection \cite{Boem8305620} tasks of LISs can be carried out successfully, so it is essential to design algorithms to accurately acquire system states. Theoretically, by joining all subsystems, the state of the whole system can be estimated by the classical Kalman filtering. However, the centralized method imposes a huge computation, storage and communication burden on the estimator due to the large number of subsystems, sensors, and parameters of LISs. In contrast, the distributed manner does not need to integrate all subsystems, but only requires each subsystem to collect information from its neighbor subsystems, and is therefore more suitable for solving the estimation problem of LISs than the centralized manner.

In this century, many distributed estimation methods have been proposed. \cite{Khan4547458} decomposed sparsely banded LISs into multiple subsystems with overlapping states and measurements and proposed a fusion method to combine overlapping information. Also, for the banded LIS, a moving horizon estimation algorithm was derived in \cite{Haber6553105} by the Chebyshev approximation method. Moreover, Kalman-type distributed estimators were proposed in \cite{CHEN2019228} and \cite{Chen9416784} for a class of sequential LISs without loops exist among subsystems. In \cite{ZHANG2023111144}, the sequential condition above was relaxed to allow for the existence of some special loops. Note that the topology of physical characteristics determine LISs and cannot be changed unilaterally. Therefore, the application scenarios of the methods in \cite{Khan4547458,Haber6553105,CHEN2019228,Chen9416784,ZHANG2023111144} are limited. In \cite{Wang9385997}, a linear minimum mean-square error distributed estimator was proposed for general LISs without topology constraints. Also for general LISs, \cite{Mu10364032} developed an optimization-based design scheme for stabilizing feedback systems using free-weighting matrix techniques. However, the methods in \cite{Wang9385997} and \cite{Mu10364032} require that the parameters of the global LIS be available to each local subsystem, which imposes a significant computational and storage burden. While a fully distributed algorithm was proposed in \cite{Zhang10480463}, its computation burden is still large since an optimization problem needs to be solved in real-time. Moreover, in \cite{Stefano1040844} and \cite{Riverso6760656}, the concept of robust positive invariance was utilized to design a distributed estimator. Nevertheless, these two works assumed that the measurement noise is zero, which is quite restrictive. \cite{FARINA2010910} and \cite{Roshany6075282} respectively designed distributed Moving horizon estimators and Kalman filters by assuming that the estimation error covariance matrices among subsystems are zero. Although this assumption facilitates the design, it may lead to inconsistent estimates and poor robustness \cite{UHLMANN2003201}. A plug-and-play distributed Kalman filter was proposed in \cite{Farina7762729}, which ensures the consistency of estimates by constructing upper bounds. Then, following the work in \cite{Farina7762729}, a more compact estimation method was proposed in \cite{FARINA2023100820} using linear matrix inequality (LMI) techniques.

It should be emphasized that although there have been a lot of works on the distributed estimation problems for LISs, they have some common limitations: i) Lack of discussion on time-varying LISs. In practice, the coupling and sampling frequency of the LIS are easily changed, which often makes the LIS time-varying \cite{Siljak1978LargeScaleDS}. However, only the time-invariant LIS is discussed in \cite{Haber6553105,ZHANG2023111144,Wang9385997,Mu10364032,FARINA2010910,Roshany6075282} and \cite{FARINA2023100820}. ii) The topology is not general. The methods in \cite{Khan4547458,Haber6553105,CHEN2019228,ZHANG2023111144,Chen9416784} are only suitable for some LISs with specific topologies, which limits their application scenarios. iii) The estimation methods are only semi-distributed. The methods in \cite{CHEN2019228,Wang9385997,Mu10364032}, and \cite{Stefano1040844} are only distributed in terms of communication, while their storage and computation are still centralized, i.e., local estimators therein depend on the global system model. iv) The communication and computation costs are relatively high. As mentioned before, the local estimators in \cite{CHEN2019228,Wang9385997,Mu10364032}, and \cite{Stefano1040844} require the global system model, and thus, their computational effort will be related to the scale of the LIS, which tends to be large. Moreover, the methods in \cite{Chen9416784,Wang9385997,Zhang10480463} and \cite{FARINA2023100820} need to solve optimization problems in real-time, which is also tedious. v) The stability analysis may need to be further refined. Although \cite{Khan4547458,Haber6553105,CHEN2019228,ZHANG2023111144,Chen9416784} provide stability analysis, they have strict topology assumptions. A necessary stability condition is derived in \cite{FARINA2023100820}, but the sufficient condition is not provided. Moreover, the convergence result in \cite{Farina7762729} is sensitive to initial conditions and requires invertible state transition matrices, which are strict in practical engineering.

In light of the preceding discussion, this paper aims to propose a fully distributed estimator for LISs without topology constraints while analyzing its stability. For the time-varying LIS, the main results are summarized as follows:
\begin{itemize}
\item A fully distributed estimator is formulated by recursively solving a distributed modified Riccati equation (DMRE) with decoupling variables.
\item By decomposing the LIS according to the block structure of the transition matrix, it is shown that the stability of each subsystem is related only to its neighbors if the decoupling variable is set to the number of its out-neighbors. Additionally, we prove that any LIS is equivalent to a Markov system in terms of the dynamic behavior, providing an explicit expression for the equivalent Markov system. Based on this, we show that the aforementioned stability decoupling can also be achieved if the decoupling variable is chosen as the number of in-neighbors.
\item Following the conclusion on stability decoupling, we prove that the distributed estimator is stable if the DMRE is uniformly bounded from above. Moreover, it is demonstrated that the uniform detectability or reachability of the local subsystem ensures the stability of the distributed estimator provided that the coupling strengths among subsystems are sufficiently small.
\end{itemize}
Additionally, for the case where the LIS is time-invariant, we further obtain the following results:
\begin{itemize}
\item The feasibility of a centralized LMI is shown to be both sufficient and necessary for bounding the DMRE through analyzing the spectral radius of a linear operator. Moreover, by seeking a feasible solution to the centralized LMI, an easy-to-verify distributed stability condition is derived.
\item By investigating the monotonicity of the DMRE, it is demonstrated that the distributed estimator converges to a unique steady state for any initial condition. Then, a steady-state distributed estimator is given based on the convergence result.
\end{itemize}
%\item In set-membership LISs, it is developed to represent the size of the Minkowski sum of ellipsoids by the sum of the covariance matrices. This avoids the conservation induced by constructing the ellipsoidal upper bound. This metric has a similar evolution to the MSE upper bound constructed in stochastic LISs. Therefore, based on the conclusions in stochastic LIS, we can directly design the optimal distributed estimator and analyze the boundedness and convergence of the developed performance metric. Since the developed performance metric is bounded does not mean that the estimate error set is bounded, an ellipsoidal upper bound is constructed to prove the boundedness of the estimate error set. The steady-state distributed estimator is also designed, and is proved to almost equal to the optimal distributed estimator in an infinite time domain.
%\item The distributed estimator designed in this paper is Kalman-type, which only requires local subsystems to transmit local estimates. In many literatures, Luenberger-type distributed estimators are employed, which require subsystems to transmit local estimates and measurements together. Then, utilizing the idea of constructing cross-correlation-independent upper bounds as well, we design a Luenberger-type distributed estimator. It is proven that the Luenberger-type distributed estimator coincides the Kalman-type distributed estimator above, which suggests that transmitting local measurements is not necessary.

{\em Notations.} ${\mathbb{R}}^r$ and ${\mathbb{R}}^{r\times s}$ denote the $r$ dimensional and $r\times s$ dimensional Euclidean spaces, respectively. $\mathrm{Diag}\{\cdot\}$ stands for block diagonal matrix. The notations $[\cdots,\cdots,\cdots]$ and $[\cdots;\cdots;\cdots]$ indicate vertical and horizontal concatenations, respectively. The symmetric terms in a symmetric matrix are denoted by ‘‘$\star$’’. $I$ stands for identity matrix. $\mathbf{1}_{rs}$ represents the all-1 matrix in ${\mathbb{R}}^{r\times s}$. $\mathrm{Tr}(\cdot)$, $\mathrm{Det}(\cdot)$ and, $\|\cdot\|$ represent the trace, determinant and 2-norm of a matrix, respectively. The inner product of two matrices is defined by $\langle X,Y\rangle\triangleq\mathrm{Tr}(Y^\top X)$. The spectral radius of a linear operator $\mathfrak{T}:\mathbb{S}\to\mathbb{S}$ on a Hilbert space $\mathbb{S}$ is denoted as $\rho(\mathfrak{T})$. For $X,Y\in{\mathbb{R}}^{r\times r}$, $X>Y$ and $X\geq Y$ mean that $X-Y$ is symmetric positive definite and semi-definite, respectively. $\sqrt{X}$ is the square root of a matrix $X$, where $X\geq 0$. $\odot$ and $\otimes$ denote Hadamard product and Kronecker product, respectively. $\delta$ is the Kronecker delta function. The power of a set is denoted as $|\cdot|$ and the power of the empty set is $0$. $\mathrm{p}(\cdot)$ and $\mathrm{E}[\cdot]$ represent the probability density function and expectation of a random variable. $\mathrm{Pr}(\cdot)$ represents the probability of a random event. The composite of functions is defined as $\mathfrak{F}\circ \mathfrak{G}(\cdot)\triangleq \mathfrak{F}\big(\mathfrak{G}(\cdot)\big)$.
%$\mathrm{Abs}(\cdot)$ denotes the element-to-element absolute value operator. An $m$-order zonotope with the center vector $c\in\mathbb{R}^n$ and the generator matrix $G\in\mathbb{R}^{n\times m}$ is defined by
%\begin{equation}
%\mathbb{Z}[c,G]\triangleq\{c+Gx|x\in\mathbb{R}^m,\|x\|_{\infty}\leq1\}.\nonumber
%\end{equation}
%Given two sets $\mathbb{S}_1\in\mathbb{R}^n$ and $\mathbb{S}_2\in\mathbb{R}^n$, their Minkowski sum is defined by
%\begin{equation}
%\mathbb{S}_1\oplus\mathbb{S}_2\triangleq\{x\in\mathbb{R}^n|x=x_1+x_2,x_1\in\mathbb{S}_1, x_2\in\mathbb{S}_2\}.\nonumber
%\end{equation}
%The covariance and F-radius of a zonotope $\mathbb{Z}[c,G]\in\mathbb{R}^n$ is $P=GG^T$ and $\sqrt{\mathrm{Tr}(P)}$, respectively.
\section{Detectability and Reachability}
This section introduces some basic concepts used in this paper \cite{Optimalfiltering,rugh1996linear,ZHANG201784,Anderson0319002,moore1980coping}.

\begin{definition}\label{DeUR}
A time-varying matrix pair $\big(X(k),Y(k)\big)$ is uniformly reachable if there exist constants $t\geq0$ and $r>0$ such that
\begin{equation}\begin{aligned}
\sum^t_{i=0}&\Phi_X(k+t+1,k+i+1)Y(k+i)\\
&\times Y^\top(k+i)\Phi_X^\top(k+t+1,k+i+1)\geq rI\nonumber
\end{aligned}\end{equation}
for all $k$, where
\begin{equation}\begin{aligned}
\Phi_X(k+t,k)\triangleq
\left\{ \begin{array}{l}
X(k-1+t)\times\cdots\times X(k),\ t\geq1,\\
I,\ t=0.
\end{array} \right.\nonumber
\end{aligned}\end{equation}
%Moreover, $\big(Y(k),X(k)\big)$ is uniformly observable if $\big(X^\top(k),Y^\top(k)\big)$ is uniformly reachable.
\end{definition}

\begin{definition}\label{DeUD}
A time-varying matrix pair $\big(Y(k),X(k)\big)$ is uniformly detectable if there exist constants $\nu\geq\mu\geq0$, $0\leq\gamma<1$ and $\sigma>0$, such that whenever
\begin{equation}\begin{aligned}
\|\Phi_X(k+\mu,k)\xi\|\geq\gamma\|\xi\|\nonumber
\end{aligned}\end{equation}
for some $\xi$ and $k$, then
\begin{equation}\begin{aligned}
\xi^\top\sum^\nu_{i=0}\Phi_X^\top(k+i)Y^\top(k+i)Y(k+i)\Phi_X(k+i,k)\xi&\\
\geq \sigma \|\xi\|^2&.\nonumber
\end{aligned}\end{equation}
\end{definition}

Additionally, if the matrix pair is time-invariant, then the uniform reachability and detectability will reduce to reachability and detectability, as shown below.

\begin{definition}\label{DeTIR}
A time-invariant matrix pair $(X,Y)$ is reachable if there exist constants $t\geq0$ and $r>0$ such that
\begin{equation}\begin{aligned}
\sum^t_{i=0}&X^iYY^\top(X^i)^\top\geq rI.\nonumber
\end{aligned}\end{equation}
\end{definition}

\begin{definition}\label{DeTID}
A time-invariant matrix pair $(Y,X)$ is detectable if there exists $K$ such that $\rho(X+KY)<1$.
\end{definition}

\begin{remark}
The (uniform) detectability is a necessary condition for stabilizing an estimator \cite{Optimalfiltering}. The (uniform) reachability serves to ensure the invertibility of a matrix, which is required in the stability analysis.
%for constructing the Lyapunov function in the proof of Theorem \ref{TheoremLya}.
\end{remark}

\section{Problem Formulation}
Consider a time-varying LIS whose global state-space model is given by
\begin{equation}
\left\{ \begin{array}{l}
x(k)=A(k-1)x(k-1)+w(k-1),\\
z(k)=C(k)x(k)+v(k),
\end{array} \right.
\label {eq:1}
\end{equation}
where $x(k)\in\mathbb{R}^{n}$ and $z(k)\in\mathbb{R}^{m}$ are respectively the system state and measurement, $w(k)$ and $v(k)$ are respectively the process noise and measurement noise, the initial state satisfies $\|x(0)-\hat{x}(0)\|<\infty$ for a known $\hat{x}(0)$, the noise sequences $\{\|w(k)\|\}_{k\geq0}$ and $\{\|v(k)\|\}_{k\geq1}$ are uniformly bounded from above, the system matrix sequences $\{\|A(k)\|\}_{k\geq0}$ and $\{\|C(k)\|\}_{k\geq1}$ are uniformly bounded from above, $\big(C(k),A(k)\big)$ is uniformly detectable. The system matrices $A(k)$ and $C(k)$ have the block structure
\begin{equation}\begin{aligned}
A(k)=\begin{bmatrix}
{A_{11}(k)} & {A_{12}(k)} & {\cdots} & {A_{1s}(k)}\\
{A_{21}(k)} & {A_{22}(k)} & {\cdots} & {A_{2s}(k)}\\
{\vdots} & {\vdots} & {\ddots} & {\vdots}\\
{A_{s1}(k)} & {A_{s2}(k)} & {\cdots} & {A_{ss}(k)}
\end{bmatrix},\nonumber
\end{aligned}\end{equation}
\begin{equation}\begin{aligned}
C(k)=\mathrm{Diag}\big(C_1(k),C_2(k),\cdots,C_s(k)\big),\nonumber
\end{aligned}\end{equation}
where $s$ is the total number of subsystems. Then, define
\begin{equation}
\mathbb{I}_i(k)\triangleq\{j|A_{ij}(k)\neq0,j=1,\cdots,s,j\neq i\},\nonumber
\end{equation}
\begin{equation}
\mathbb{O}_i(k)\triangleq\{j|A_{ji}(k)\neq0,j=1,\cdots,s,j\neq i\},\nonumber
\end{equation}
where $\mathbb{I}_i(k)$ and $\mathbb{O}_i(k)$ represent the index set of in-neighbors and out-neighbors of $i$th subsystem, respectively. With the setting above, the global system (1) can be reformulated as the distributed form
\begin{equation}
\left\{ \begin{array}{l}
x_{i}(k)=\sum_{j\in\mathbb{I}_i(k-1)\cup\{i\}}A_{ij}(k-1)x_{j}(k-1)\\
\ \ \ \ \ \ \ \ \ \ +w_{i}(k-1),\\
z_{i}(k)=C_i(k)x_{i}(k)+v_i(k),
\end{array} \right.
\end{equation}
where $x_i(k)\in\mathbb{R}^{n_i}$ and $z_i(k)\in\mathbb{R}^{m_i}$ are respectively the subsystem state and subsystem measurement, $w_i(k)\in\mathbb{R}^{n_i}$ and $v_i(k)\in\mathbb{R}^{m_i}$ are respectively the subsystem process noise and subsystem measurement noise, $x(k)=[x_1(k);x_2(k);\cdots;x_s(k)]$ and $z(k)=[z_1(k);z_2(k);\cdots;z_s(k)]$. Evidently, $\sum^{s}_{i=1}n_i=n$ and $\sum^{s}_{i=1}m_i=m$.

In each subsystem, a local estimator is employed, i.e.,
\begin{equation}\begin{aligned}
\hat{x}_i(k)=\bar{x}_i(k)+K_{i}(k)\big(z_i(k)-C_i(k)\bar{x}_i(k)\big),
\end{aligned}\end{equation}
where $\bar{x}_i(k)=\sum_{j\in\mathbb{I}_i(k-1)\cup\{i\}}A_{ij}(k-1)\hat{x}_j(k-1)$ and $K_i(k)$ is the local gain. Inspired by the classical Kalman filtering theory\cite{Optimalfiltering}, the local estimator gain $K_{i}(k)$ is designed as the Kalman-like form
\begin{equation}\begin{aligned}
K_i(k)=\bar{P}_i(k)C_i^\top(k)\big(C_i(k)\bar{P}_i(k)C_i^\top(k)+R_i(k)\big)^{-1},
\end{aligned}\end{equation}
where $\bar{P}_i(k)$ is obtained by solving the DMRE
\begin{equation}\begin{aligned}
P_i(k)=&\bar{P}_i(k)-\bar{P}_i(k)C_i^\top(k)\\
&\times\big(C_i(k)\bar{P}_i(k)C_i^\top(k)+R_i(k)\big)^{-1}C_i(k)\bar{P}_i(k),
\end{aligned}\end{equation}
\begin{equation}\begin{aligned}
\bar{P}_i(k+1)=&\vartheta_i^2(k)\sum_{j\in\mathbb{I}_i(k)\cup\{i\}}A_{ij}(k)P_j(k)A_{ij}^\top(k)\\
&+Q_i(k),
\end{aligned}\end{equation}
where $\vartheta_i(k)$, $Q_i(k)\geq0$ and $R_i(k)>0$ are adjustable parameters. The objective of this paper is to analyze the stability of the distributed estimator (3) with local gain (4).
\section{Stability Analysis for The Time-Varying Case}
In this section, we discuss how to guarantee the stability of the distributed estimator (3) with the local gain (4). For ease of presentation, we define
\begin{equation}
\begin{aligned}
\left\{ \begin{array}{l}
\bar{P}(k)\triangleq\mathrm{Diag}\big(\bar{P}_1(k),\bar{P}_2(k),\cdots,\bar{P}_s(k)\big),\\
P(k)\triangleq\mathrm{Diag}\big(P_1(k),P_2(k),\cdots,P_s(k)\big),\\
K(k)\triangleq\mathrm{Diag}\big(K_1(k),K_2(k),\cdots,K_s(k)\big),\\
Q(k)\triangleq\mathrm{Diag}\big(Q_1(k),Q_2(k),\cdots,Q_s(k)\big),\\
R(k)\triangleq\mathrm{Diag}\big(R_1(k),R_2(k),\cdots,R_s(k)\big),\\
\hat{x}(k)\triangleq[\hat{x}_1(k);\hat{x}_2(k);\cdots;\hat{x}_s(k)],\\
\bar{x}(k)\triangleq[\bar{x}_1(k);\bar{x}_2(k);\cdots;\bar{x}_s(k)],\\
\hat{e}(k)\triangleq[\hat{e}_1(k);\hat{e}_2(k);\cdots;\hat{e}_s(k)],\\
\bar{e}(k)\triangleq[\bar{e}_1(k);\bar{e}_2(k);\cdots;\bar{e}_s(k)],\\
\hat{e}_i(k)\triangleq x_i(k)-\hat{x}_i(k),\\
\bar{e}_i(k)\triangleq x_i(k)-\bar{x}_i(k).
\end{array} \right.
\end{aligned}
\end{equation}
Additionally, the following conditions concerning the adjustable parameters $\vartheta_i(k)$, $Q_i(k)$ and $R_i(k)$ will be frequently invoked in the stability analysis.

\begin{condition}\label{AsLUR}
For all $i=1,\cdots,s$, $\big(A_{ii}(k),\sqrt{Q_i(k)}\big)$ is uniformly reachable.
\end{condition}

\begin{condition}\label{AsQRB}
For all $i=1,\cdots,s$, there exist $q_{u,i}\geq0$ and $r_{u,i}\geq r_{l,i}>0$ such that $0\leq Q_i(k)\leq q_{u,i}I$ and $r_{l,i}I\leq R_i(k+1)\leq r_{u,i}I$ for all $k\geq0$.
\end{condition}

\begin{condition}\label{AsPB}
For all $i=1,\cdots,s$, there exists $p_{u,i}\geq0$ such that $P_i(k)\leq p_{u,i} I$ for all $k\geq0$.
\end{condition}

\begin{condition}\label{AsThetaB}
For all $i=1,\cdots,s$, there exist $\vartheta_{l,i}>0$ and $\vartheta_{u,i}>0$ such that $\vartheta_{l,i}\leq\|\vartheta_i(k)\|\leq\vartheta_{u,i}$ for all $k\geq 0$.
\end{condition}

\begin{remark}
Note that Conditions \ref{AsLUR}, \ref{AsQRB}, and \ref{AsThetaB} can be easily fulfilled by choosing parameters $\vartheta_i(k)$, $Q_i(k)$, and $R_i(k)$ in a distributed manner. Moreover, we will also analyze the circumstances under which Condition \ref{AsPB} can be fulfilled.
\end{remark}
\subsection{Stability of the estimate error system: A Lyapunov analysis method}
The following proposition proves the uniform detectability of a matrix pair, which is crucial in the stability analysis of the estimate error system.

\begin{proposition}\label{PrUD}
If the parameters $Q_i(k)$, $R_i(k)$, and $\vartheta_i(k)$ are chosen such that Conditions \ref{AsQRB}, \ref{AsPB}, and \ref{AsThetaB} are satisfied, then the time-varying matrix pair $\big(C(k)\Phi_{\hat{A}}(k,k-1),\Phi_{\hat{A}}(k,k-1)\big)$ is uniformly detectable, where
\begin{equation}\begin{aligned}
\hat{A}(k)\triangleq\big(I-K(k+1)C(k+1)\big)A(k).\nonumber
\end{aligned}\end{equation}
\end{proposition}

\begin{proof}
Let $\mu$, $\nu$, $\gamma$, and $\sigma$ be the parameters associated with the uniform detectability of $\big(C(k),A(k)\big)$. Define
\begin{equation}\begin{aligned}
\hat{\mathcal{O}}(k+\nu,k)\triangleq
\begin{bmatrix}
{C(k)}\Phi_{\hat{A}}(k,k)\\
{C(k+1)\Phi_{\hat{A}}(k+1,k)}\\
{\vdots}\\
{C(k+\nu)\Phi_{\hat{A}}(k+\nu,k)}\\
\end{bmatrix},\nonumber
\end{aligned}\end{equation}
\begin{equation}\begin{aligned}
\mathcal{O}(k+\nu,k)
\triangleq
\begin{bmatrix}
{C(k)\Phi_A(k,k)}\\
{C(k+1)\Phi_A(k+1,k)}\\
{\vdots}\\
{C(k+\nu)\Phi_A(k+\nu,k)}\\
\end{bmatrix}.\nonumber
\end{aligned}\end{equation}
Then, one can reformulate $\hat{\mathcal{O}}(k+\nu,k)$ as
\begin{equation}\begin{aligned}
\hat{\mathcal{O}}(k+\nu,k)=\hat{W}(k+\nu)\mathcal{O}(k+\nu,k),
\end{aligned}\end{equation}
where
\begin{equation}\begin{aligned}
\hat{W}(k+\nu)\triangleq\begin{bmatrix}
{I} & {0} & {0} & {\cdots} & {0}\\
{0} & {\hat{L}_{11}(k)} & {0} & {\cdots} & {0}\\
{0} & {\hat{L}_{21}(k)} & {\hat{L}_{22}(k)} & {\ddots} & {\vdots}\\
{\vdots} & {\vdots} & {\vdots} & {\ddots} & {0}\\
{0} & {\hat{L}_{\nu1}(k)} & {\hat{L}_{\nu2}(k)} & {\cdots} & {\hat{L}_{\nu\nu}(k)}
\end{bmatrix},\nonumber
\end{aligned}\end{equation}
\begin{equation}\begin{aligned}
&\hat{L}_{ij}(k)\triangleq
\left\{ \begin{array}{l}
I-C(k+i)K(k+i),\ i=j,\\
-C(k+i)\Phi_{\hat{A}}(k+i,k+j)K(k+j),\ i>j.
\end{array} \right.\nonumber
\end{aligned}\end{equation}
According to Conditions \ref{AsQRB}, \ref{AsPB}, and \ref{AsThetaB}, $\{\|\hat{L}_{ij}(k)\|\}_{k\geq0}$ is uniformly bounded from above. In this case, there exists a constant $w_u>0$ such that
\begin{equation}\begin{aligned}
\hat{W}^\top(k+\nu)\hat{W}(k+\nu)\leq w_u I,\ \forall k\geq 0.
\end{aligned}\end{equation}
Moreover, it is clear from (5) and (6) that the sequence $\{\bar{P}(k)\}_{k\geq 1}$ is uniformly bounded from above if Conditions \ref{AsQRB}, \ref{AsPB}, and \ref{AsThetaB} are fulfilled. With this boundedness result, one can derive that there exists a constant $w>0$ such that
\begin{equation}\begin{aligned}
&\mathrm{Det}\big(\hat{W}(k+\nu)\big)\\
%=&\prod^{\nu}_{i=2}\mathrm{Det}\big(I-C(k+i)K(k+i)\big)\\
=&\prod^{\nu}_{i=1}\frac{\mathrm{Det}\big(R(k+i)\big)}{\mathrm{Det}\big(C(k+i)\bar{P}(k+i)C^\top(k+i)+R(k+i)\big)}\\
\geq& w,\ \forall k\geq 0.
\end{aligned}\end{equation}
Then, it follows from (9) and (10) that there exists a constant $w_l>0$ such that
\begin{equation}\begin{aligned}
\hat{W}^\top(k+\nu)\hat{W}(k+\nu)\geq w_l I,\ \forall k\geq0.
\end{aligned}\end{equation}
Based on (8), (9), and (11), one has
\begin{equation}\begin{aligned}
o_l\mathcal{O}^\top(k+\nu,&k)\mathcal{O}(k+\nu,k)\\
\leq\hat{\mathcal{O}}^\top&(k+\nu,k)\hat{\mathcal{O}}(k+\nu,k)\\
&\leq o_u\mathcal{O}^\top(k+\nu,k)\mathcal{O}(k+\nu,k)
\end{aligned}\end{equation}
for all $k\geq0$ and some constants $0<o_l\leq o_u$.

Additionally, note that the closed-loop matrix $\Phi_{\hat{A}}(k+\mu,k)$ can be reformulated as
\begin{align}
\Phi_{\hat{A}}(k+\mu,k)=\Phi_A(k+\mu,k)+&[0,L_1(k),L_2(k),\cdots,L_\mu(k)]\nonumber\\
&\times\mathcal{O}(k+\mu,k),
\end{align}
where
\begin{equation}\begin{aligned}
L_i(k)\triangleq-\Phi_{\hat{A}}(k+\mu,k+i)K(k+i).\nonumber
\end{aligned}\end{equation}
%\begin{equation}\begin{aligned}
%\Phi_{\hat{A}}(k+\mu,k)
%=&\Phi(k+\mu,k)+[-\Phi_{\hat{A}}(k+\mu,k+1)K(k+1)]
%\begin{bmatrix}
%{C(k+1)A(k)}\\
%{C(k+2)A(k+1)A(k)}\\
%{\vdots}\\
%{C(k+\mu)A(k+\mu-1)\times\cdots\times A(k)}\\
%\end{bmatrix}\\
%=&\Phi(k+\mu,k)+[\star]
%\begin{bmatrix}
%{C(k)}\\
%{C(k+1)A(k)}\\
%{C(k+2)A(k+1)A(k)}\\
%{\vdots}\\
%{C(k+\mu)A(k+\mu-1)\times\cdots\times A(k)}\\
%\end{bmatrix}\\
%=&\Phi_A(k+\mu,k)+[\star]\mathcal{O}(k+\mu,k),
%\end{aligned}\end{equation}
Under Conditions \ref{AsQRB}, \ref{AsPB}, and \ref{AsThetaB}, it is trivial that
\begin{equation}\begin{aligned}
\|[0,L_1(k),L_2(k),\cdots,L_\mu(k)]\|\leq\alpha\nonumber
\end{aligned}\end{equation}
for some $\alpha\geq0$ and all $k\geq0$. Thus, for any $\xi\in\mathbb{R}^{n}$, it follows from (13) that
\begin{equation}\begin{aligned}
\|\Phi_{\hat{A}}(k+\mu,k)\xi\|\leq&\|\Phi_A(k+\mu,k)\xi\|\\
&+\alpha\|\mathcal{O}(k+\mu,k)\xi\|.
\end{aligned}\end{equation}

Let $\bar{\sigma}$ satisfy $0<\bar{\sigma}\leq o_l\sigma$. If $\xi^\top\hat{\mathcal{O}}^\top(k+\nu,k)\hat{\mathcal{O}}(k+\nu,k)\xi<\bar{\sigma}\xi^\top\xi$, one can derive from (12) that
\begin{equation}\begin{aligned}
\xi^\top\mathcal{O}^\top(k+\nu,k)\mathcal{O}(k+\nu,k)\xi<\frac{\bar{\sigma}}{o_l}\xi^\top\xi\leq \sigma\xi^\top\xi.\nonumber
\end{aligned}\end{equation}
This inequality implies (recall Definition \ref{DeUD})
\begin{equation}\begin{aligned}
\|\Phi_A(k+\mu,k)\xi\|<\gamma\|\xi\|,
\end{aligned}\end{equation}
\begin{equation}\begin{aligned}
\|\mathcal{O}(k+\mu,k)\xi\|&\leq\big(\xi^\top\mathcal{O}^\top(k+\nu,k)\mathcal{O}(k+\nu,k)\xi\big)^{1/2}\\
&<\sqrt{\frac{\bar{\sigma}}{o_l}}\|\xi\|.
\end{aligned}\end{equation}
Substituting (15) and (16) into (14) yields
\begin{equation}\begin{aligned}
\|\Phi_{\hat{A}}(k+\mu,k)\xi\|<(\gamma+\alpha\sqrt{\frac{\bar{\sigma}}{o_l}})\|\xi\|=\bar{\gamma}\|\xi\|,\nonumber
\end{aligned}\end{equation}
where $\bar{\gamma}=\gamma+\alpha\sqrt{\frac{\bar{\sigma}}{o_l}}$. Since $\gamma<1$, one can choose a small $\bar{\sigma}$ such that $0<\bar{\gamma}<1$. Thus, $\big(C(k),\Phi_{\hat{A}}(k+1,k)\big)$ is uniformly detectable with parameters $\mu$, $\nu$, $\bar{\gamma}$ and $\bar{\sigma}$. Then, it follows from Lemma \ref{LemmaA1} in Appendix that $\big(C(k)\Phi_{\hat{A}}(k,k-1),\Phi_{\hat{A}}(k,k-1)\big)$ is uniformly detectable. The proof is completed.
\end{proof}

Based on Proposition \ref{PrUD}, it is ready to give the following theorem.

\begin{theorem}\label{TheoremLya} Consider the LIS (2) and the distributed estimator (3) with gain (4). If $\vartheta_i(k)=\sqrt{|\mathbb{O}_i(k+1)|+1}$ for $i=1,2,\cdots,s$ and $k\geq0$ and Conditions \ref{AsLUR}, \ref{AsQRB}, and \ref{AsPB} are fulfilled, then the estimate error system is exponentially stable, i.e.,
\begin{equation}\begin{aligned}
\|\hat{e}(k)\|\leq ab^k+c,
\end{aligned}\end{equation}
where $a\geq0$, $0\leq b<1$ and $c\geq0$ are constants. Meanwhile, if $w(k)=0$ and $v(k+1)=0$ for all $k\geq0$, then
\begin{equation}\begin{aligned}
\|\hat{e}(k)\|\leq ab^k.
\end{aligned}\end{equation}
\end{theorem}

\begin{proof}
Step 1. Prove that there exist constants $p_l>0$ and $t\geq0$ such that $P(k)\geq p_lI$ for all $k\geq t$. Recall the evolution of $\bar{P}_i(k)$ in (6), one can reformulate it as
\begin{equation}\begin{aligned}
\bar{P}_i(k+1)\geq \bar{F}_{ii}(k)\bar{P}_i(k)\bar{F}_{ii}^\top(k)+B_{ii}(k)B_{ii}^\top(k),
\end{aligned}\end{equation}
where
\begin{equation}\begin{aligned}
&\bar{F}_{ii}(k)\triangleq \vartheta_i(k)A_{ii}(k)-B_{ii}(k)G_i(k),\nonumber
\end{aligned}\end{equation}
\begin{equation}\begin{aligned}
B_{ii}(k)\triangleq[\vartheta_i(k)A_{ii}(k)K_i(k)\sqrt{R_i(k)},\sqrt{Q_i(k)}],\nonumber
\end{aligned}\end{equation}
\begin{equation}\begin{aligned}
G_i(k)\triangleq[\sqrt{R_i^{-1}(k)}C_i(k);0].\nonumber
\end{aligned}\end{equation}
%\begin{equation}\begin{aligned}
%&\bar{A}_i(k)\triangleq A_{ii}(k)-B_{i}(k)G_i(k),\nonumber
%\end{aligned}\end{equation}
%\begin{equation}\begin{aligned}
%B_{i}(k)\triangleq[A_{ii}(k)K_i(k)\sqrt{R_i(k)},\sqrt{Q_i(k)}],\nonumber
%\end{aligned}\end{equation}
%\begin{equation}\begin{aligned}
%G_i(k)\triangleq[\sqrt{R_i^{-1}(k)}C_i(k);0].\nonumber
%\end{aligned}\end{equation}
Under Condition \ref{AsLUR}, it is trivial that the time-varying matrix pair $\big(\vartheta_i(k)A_{ii}(k),B_{ii}(k)\big)$ is uniformly reachable. Moreover, the iteration of (19) yields
\begin{equation}\begin{aligned}
\bar{P}_i(k+t_i)\geq&\sum^{t_i-1}_{j=0}\Phi_{\bar{A}_{ii}}(k+t_i,k+j+1)B_{ii}(k+j)\\
&\ \ \ \ \ \ \times B_{ii}^\top(k+j)\Phi_{\bar{A}_{ii}}^\top(k+t_i,k+j+1).\nonumber
\end{aligned}\end{equation}
According to Condition \ref{AsQRB}, one knows that the sequence $\{\|G_i(k)\|\}_{k\geq1}$ is uniformly bounded from above. Thus, it follows from the invariance of reachability under bounded feedback \cite{Anderson0319002} that
\begin{equation}\begin{aligned}
\bar{P}_i(k)\geq \bar{p}_{l,i}I,\ \forall k\geq t_i,\nonumber
\end{aligned}\end{equation}
where $t_i\geq1$ and $\bar{p}_{l,i}>0$ are constants. Then, utilizing the matrix inverse lemma for (5) yields
\begin{equation}\begin{aligned}
P_i(k)&=\big(\bar{P}_i^{-1}(k)+C_i^\top(k)R_i^{-1}(k)C_i(k)\big)^{-1}\\
&\geq p_{l,i}I,\ \forall k\geq t_i,\nonumber
\end{aligned}\end{equation}
where $p_{l,i}>0$ is a constant. Consequently, $P(k)\geq p_{l} I$ for all $k\geq t$, where $p_l=\min\{p_{l,1},p_{l,2},\cdots,p_{l,s}\}$ and $t=\max\{t_1,t_2,\cdots,t_s\}$.

Step 2. Construct the Lyapunov function. According to (2) and (3), the local estimate error system is given by
\begin{equation}\begin{aligned}
\hat{e}_i(k)=&\big(I-K_i(k)C_i(k)\big)\\
&\times\sum_{j\in\mathbb{I}_i(k-1)\cup\{i\}}A_{ij}(k-1)\hat{e}_{j}(k-1)\\
&+\big(I-K_i(k)C_i(k)\big)w_i(k-1)\\
&-K_i(k)v_i(k).
\end{aligned}\end{equation}
With (20), the global estimate error system can be given by
\begin{equation}\begin{aligned}
\hat{e}(k)=&\big(I-K(k)C(k)\big)A(k-1)\hat{e}(k-1)\\
&+\big(I-K(k)C(k)\big)w(k-1)-K(k)v(k).
\end{aligned}\end{equation}
Consider the noise-free version of the system (21), that is,
\begin{equation}\begin{aligned}
\eta(k)=\big(I-K(k)C(k)\big)A(k-1)\eta(k-1).
\end{aligned}\end{equation}
Under Conditions \ref{AsQRB} and \ref{AsPB}, one knows that the noise sequence $\{\|\big(I-K(k)C(k)\big)w(k-1)-K(k)v(k)\|\}_{k\geq1}$ in (21) is uniformly bounded from above. Thus, if the system (22) is exponentially stable, then the inequalities (17) and (18) hold.

Choose the Lyapunov candidate function as
\begin{equation}\begin{aligned}
V(k)=\eta^\top(k)\big(\Lambda(k)P^{-1}(k)+\epsilon I\big)\eta(k),\ k\geq t,
\end{aligned}\end{equation}
where $\epsilon>0$ is a constant to be determined and
\begin{equation}\begin{aligned}
\Lambda(k)\triangleq\mathrm{Diag}\big((|\mathbb{O}_1(k)|+1)I,\cdots,(|\mathbb{O}_s(k)|+1)I\big).\nonumber
\end{aligned}\end{equation}
Without loss of generality, we assume $k>t$ in the subsequent proof.

Step 3. Prove that the Lyapunov function (23) decays exponentially fast. According to the Kalman filtering theory \cite{Optimalfiltering}, one can reformulate the Kalman gain $K_i(k)$ as $K_i(k)=P_i(k)C_i^\top(k)R_i^{-1}(k)$. Substituting the reformulation into (5) yields
\begin{equation}\begin{aligned}
&P_i(k)-K_i(k)R_i(k)K_i^\top(k)\\
%=&P_i(k)-P_i(k)C_i^T(k)R_i^{-1}(k)C_i(k)P_i(k)\\
\leq&P_i(k)-P_i(k)C_i^\top(k)\big(R_i(k)+C_i(k)P_i(k)C_i^\top(k)\big)^{-1}\\
&\times C_i(k)P_i(k)\\
=&\big(P_i^{-1}(k)+C_i^\top(k)R_i^{-1}(k)C_i(k)\big)^{-1}.
\end{aligned}\end{equation}
Note that $P_i(k)$ can be reformulated as
\begin{equation}\begin{aligned}
P_i(k)=&\big(I-K_{i}(k)C_{i}(k)\big)\bar{P}_i(k)\big(I-K_{i}(k)C_{i}(k)\big)^\top\\
&+K_i(k)R_i(k)K_i^\top(k).
\end{aligned}\end{equation}
Then, it follows from (6) and (25) that
\begin{align}
&P_i(k)-K_i(k)R_i(k)K^\top_i(k)\nonumber\\
%=&\vartheta_i^2(k-1)\big(I-K_{i}(k)C_{i}(k)\big)\sum_{j\in\mathbb{I}_i(k-1)\cup\{i\}}A_{ij}(k-1)\\
%&\times P_j(k-1)A_{ij}^T(k-1)\big(I-K_{i}(k)C_{i}(k)\big)^T\\
%&+\big(I-K_{i}(k)C_{i}(k)\big)\Theta_i(k)\big(I-K_{i}(k)C_{i}(k)\big)^T\\
\geq&\vartheta_i^2(k-1)\big(I-K_{i}(k)C_{i}(k)\big)\sum_{j\in\mathbb{I}_i(k-1)\cup\{i\}}A_{ij}(k-1)\nonumber\\
&\times P_j(k-1)A_{ij}^\top(k-1)\big(I-K_{i}(k)C_{i}(k)\big)^\top.
\end{align}
Based on (24) and (26), one can derive that
\begin{equation}\begin{aligned}
&\big(P_i^{-1}(k)+C_i^\top(k)R_i^{-1}(k)C_i(k)\big)^{-1}\\
%\geq&\vartheta_i^2(k-1)\big(I-K_{i}(k)C_{i}(k)\big)\sum_{j\in\mathbb{I}_i(k-1)\cup\{i\}}A_{ij}(k-1)P_j(k-1)A_{ij}^T(k-1)\big(I-K_{i}(k)C_{i}(k)\big)^T\\
\geq&\vartheta_i^2(k-1)\big(I-K_{i}(k)C_{i}(k)\big)A_{r,i}(k-1)P(k-1)\\
&\times A_{r,i}^\top(k-1)\big(I-K_{i}(k)C_{i}(k)\big)^\top,
\end{aligned}\end{equation}
where
\begin{equation}\begin{aligned}
A_{r,i}(k)\triangleq[A_{i1}(k),A_{i2}(k),\cdots,A_{is}(k)].\nonumber
\end{aligned}\end{equation}
Then, utilizing Schur complement lemma \cite{boyd1994linear} for the inequality (27) yields
\begin{align}
&P(k-1)-\vartheta_i^2(k-1)P(k-1)A_{r,i}^\top(k-1)\nonumber\\
&\times\big(I-K_{i}(k)C_{i}(k)\big)^\top\big(P_i^{-1}(k)+C_i^\top(k)R_i^{-1}(k)C_i(k)\big)\nonumber\\
&\times\big(I-K_{i}(k)C_{i}(k)\big)A_{r,i}(k-1)P(k-1)\geq0.
\end{align}

Define
\begin{equation}\begin{aligned}
\lambda_{ij}(k)\triangleq
\left\{ \begin{array}{l}
0,\ A_{ij}(k)=0,\\
1,\ A_{ij}(k)\neq0,
\end{array} \right.\nonumber
\end{aligned}\end{equation}
\begin{equation}\begin{aligned}
\Lambda_i(k)\triangleq\mathrm{Diag}\big(\lambda_{i1}(k)I,\lambda_{i2}(k)I,\cdots,\lambda_{is}(k)I\big).\nonumber
\end{aligned}\end{equation}
It is trivial that
\begin{equation}\begin{aligned}
\Lambda_i(k)P^{-1}(k)&=\Lambda_i(k)P^{-1}(k)\Lambda_i(k),\\
A_{r,i}(k)\Lambda_i(k)&=A_{r,i}(k).
\end{aligned}\end{equation}
Then, it follows from (28) and (29) that
\begin{equation}\begin{aligned}
&\Lambda_i(k-1)P^{-1}(k-1)\\
&-\vartheta_i^2(k-1)A_{r,i}^\top(k-1)\big(I-K_{i}(k)C_{i}(k)\big)^\top\\
&\times\big(P_i^{-1}(k)+C_i^\top(k)R_i^{-1}(k)C_i(k)\big)\\
&\times\big(I-K_{i}(k)C_{i}(k-1)\big)A_{r,i}(k-1)\geq0.
\end{aligned}\end{equation}
Summing (30) over the subscript $i$ yields
\begin{align}
&\Lambda(k-1)P^{-1}(k-1)-A^\top(k-1)\big(I-K(k)C(k)\big)^\top\nonumber\\
&\times\Theta(k-1)\big(P^{-1}(k)+C^\top(k)R^{-1}(k)C(k)\big)\nonumber\\
&\times\big(I-K(k)C(k)\big)A(k-1)\geq0,
\end{align}
where
\begin{equation}\begin{aligned}
\Theta(k)\triangleq\mathrm{Diag}\big(\vartheta_1^2(k)I,\vartheta_1^2(k)I,\cdots,\vartheta_s^2(k)I\big).\nonumber
\end{aligned}\end{equation}
Since $\vartheta_i(k)=\sqrt{|\mathbb{O}_i(k+1)|+1}$, one can obtain
\begin{equation}\begin{aligned}
&\Lambda(k-1)P^{-1}(k-1)\\
&-A^\top(k-1)\big(I-K(k)C(k)\big)^\top\\
&\times\Lambda(k)\big(P^{-1}(k)+C^\top(k)R^{-1}(k)C(k)\big)\\
&\times\big(I-K(k)C(k)\big)A(k-1)\geq0.
\end{aligned}\end{equation}

It follows from (32) that
\begin{equation}\begin{aligned}
&V(k+i-1)-V(k+i)\\
\geq &\eta^\top(k)\Phi_{\hat{A}}^\top(k+i,k)C^\top(k+i)R^{-1}(k+i)\\
&\times C(k+i)\Phi_{\hat{A}}(k+i,k)\eta(k)\\
&+\epsilon\eta^\top(k+i-1)\eta(k+i-1)\\
&-\epsilon\eta^\top(k+i)\eta(k+i).\nonumber
\end{aligned}\end{equation}
%\begin{equation}\begin{aligned}
%&V(k+\iota-1)-V(k+\iota)\\
%\geq &\zeta^T(k+\iota)C^T(k+\iota)R^{-1}(k+\iota)C(k+\iota)\zeta(k+\iota)+\epsilon\zeta^T(k+\iota-1)\zeta(k+\iota-1)-\epsilon\zeta^T(k+\iota)\zeta(k+\iota)\\
%=&\zeta^T(k)\hat{\Phi}^T(k+\iota,k)C^T(k+\iota)R^{-1}(k+\iota)C(k+\iota)\hat{\Phi}(k+\iota,k)\zeta(k)\\
%&+\epsilon\zeta^T(k+\iota-1)\zeta(k+\iota-1)-\epsilon\zeta^T(k+\iota)\zeta(k+\iota).
%\end{aligned}\end{equation}
Under Condition \ref{AsQRB}, one can recursively derive that
\begin{equation}\begin{aligned}
&V(k)-V(k+i)\\
\geq&\frac{1}{r_u}\eta^\top(k)\breve{\mathcal{O}}^\top(k+i,k+1)\breve{\mathcal{O}}(k+i,k+1)\eta(k)\\
&-\epsilon\eta^\top(k)\Phi_{\hat{A}}^\top(k+i,k)\Phi_{\hat{A}}(k+i,k)\eta(k)\\
&+\epsilon\eta^\top(k)\eta(k),
\end{aligned}\end{equation}
%\begin{equation}\begin{aligned}
%V(k)-V(k+\iota)\geq&\frac{1}{r_u}\sum^\iota_{i=1}\zeta^T(k)\hat{\Phi}^T(k+i,k)C^T(k+i)C(k+i)\hat{\Phi}(k+i,k)\zeta(k)\\
%&+\epsilon\zeta^T(k)\zeta(k)-\epsilon\zeta^T(k+\iota)\zeta(k+\iota)\\
%=&\frac{1}{r_u}\sum^\iota_{i=1}\zeta^T(k)\hat{\Phi}^T(k+i,k)C^T(k+i)C(k+i)\hat{\Phi}(k+i,k)\zeta(k)\\
%&+\epsilon\zeta^T(k)\zeta(k)-\epsilon\zeta^T(k)\hat{\Phi}^T(k+\iota,k)\hat{\Phi}(k+\iota,k)\zeta(k)\\
%=&\frac{1}{r_u}\zeta^T(k)\breve{\mathcal{O}}^T(k+\iota,k+1)\breve{\mathcal{O}}(k+\iota,k+1)\zeta(k)\\
%&+\epsilon\zeta^T(k)\zeta(k)-\epsilon\zeta^T(k)\hat{\Phi}^T(k+\iota,k)\hat{\Phi}(k+\iota,k)\zeta(k),
%\end{aligned}\end{equation}
where $r_u\triangleq\max\{r_{u,1},r_{u,2},\cdots,r_{u,s}\}$ and
\begin{equation}\begin{aligned}
\breve{\mathcal{O}}(k+i,k+1)\triangleq
\begin{bmatrix}
{C(k+1)}\Phi_{\hat{A}}(k+1,k)\\
{C(k+2)\Phi_{\hat{A}}(k+2,k)}\\
{\vdots}\\
{C(k+i)\Phi_{\hat{A}}(k+i,k)}\\
\end{bmatrix}.\nonumber
\end{aligned}\end{equation}
According to Proposition \ref{PrUD}, one knows that $\big(C(k)\Phi_{\hat{A}}(k,k-1),\Phi_{\hat{A}}(k,k-1)\big)$ is uniformly detectable. Let $\mu$, $\nu$, $\gamma$, and $\sigma$ be the parameters associated with the uniform detectability of $\big(C(k)\Phi_{\hat{A}}(k,k-1),\Phi_{\hat{A}}(k,k-1)\big)$. On one hand, if $\|\Phi_{\hat{A}}(k+\mu,k)\eta(k)\|\geq \gamma\|\eta(k)\|$, it follows from (33) that
\begin{equation}\begin{aligned}
V(k)-V(k+\nu)\geq(\frac{\sigma}{r_u}+\epsilon-\epsilon \beta)\eta^\top(k)\eta(k),\nonumber
\end{aligned}\end{equation}
%\begin{equation}\begin{aligned}
%V(k)-V(k+\mu)\geq&\frac{1}{r_u}\zeta^T(k)\breve{\mathcal{O}}^T(k+\mu,k+1)\breve{\mathcal{O}}(k+\mu,k+1)\zeta(k)\\
%&+\epsilon\zeta^T(k)\zeta(k)-\epsilon\zeta^T(k)\hat{\Phi}^T(k+\mu,k)\hat{\Phi}(k+\mu,k)\zeta(k)\\
%\geq&\frac{\sigma}{r_u}\zeta^T(k)\zeta(k)+\epsilon\zeta^T(k)\zeta(k)-\epsilon\zeta^T(k)\hat{\Phi}^T(k+\mu,k)\hat{\Phi}(k+\mu,k)\zeta(k)\\
%\geq&(\frac{\sigma}{r_u}+\epsilon-\epsilon \phi)\zeta^T(k)\zeta(k),
%\end{aligned}\end{equation}
where $\beta\geq 0$ is a constant satisfying $\Phi_{\hat{A}}^\top(k+\nu,k)\Phi_{\hat{A}}(k+\nu,k)\leq \beta I$ for all $k\geq0$. This constant must exist under Conditions \ref{AsQRB} and \ref{AsPB}. On the other hand, if $\|\Phi_{\hat{A}}(k+\mu,k)\zeta(k)\|<\gamma\|\zeta(k)\|$, one can derive that
\begin{equation}\begin{aligned}
V(k)-V(k+\mu)\geq\epsilon\eta^\top(k)\eta(k)-\epsilon \gamma^2\eta^\top(k)\eta(k).\nonumber
\end{aligned}\end{equation}
%\begin{equation}\begin{aligned}
%V(k)-V(k+\nu)\geq&\frac{1}{r_u}\zeta^T(k)\breve{\mathcal{O}}^T(k+\nu,k+1)\breve{\mathcal{O}}(k+\nu,k+1)\zeta(k)\\
%&+\epsilon\zeta^T(k)\zeta(k)-\epsilon\zeta^T(k)\hat{\Phi}^T(k+\nu,k)\hat{\Phi}(k+\nu,k)\zeta(k)\\
%\geq&\epsilon\zeta^T(k)\zeta(k)-\epsilon \gamma^2\zeta^T(k)\zeta(k)
%\end{aligned}\end{equation}
Clearly, one can choose a small $\epsilon$ such that $\frac{\sigma}{r_u}+\epsilon-\epsilon \beta>0$ and $\epsilon-\epsilon \gamma^2>0$. Meanwhile, note that $P^{-1}(k)\leq \frac{1}{p_l}I$ and $\Lambda(k)\leq sI$. This implies that $\eta^\top(k)\eta(k)\geq \alpha_1 V(k)$ for some constant $0<\alpha_1\leq1$. Consequently, for any $k\geq t+1$, there exists a constant $0\leq\alpha_2<1$ such that
\begin{equation}\begin{aligned}
V(k+i_k)\leq \alpha_2V(k),\nonumber
\end{aligned}\end{equation}
where $i_k\in\{\mu,\nu\}$. Thus, there exists a subsequence $\{V(k_i)\}_{i\geq1}$ decays exponentially fast with $\mu\leq k_i-k_{i-1}\leq \nu$ and $k_1\geq \max\{\mu,t+1\}$.

Step 4. Prove the exponential stability of the noise-free system (22). Under Condition \ref{AsPB}, the subsequence $\{V(k_i)\}_{i\geq1}$ decaying exponentially fast implies that the subsequence $\{\|\eta(k_i)\|\}_{i\geq1}$ also decays exponentially fast. Then, for any $k\geq \max\{\nu,t+1\}$, one can find $i\geq1$ such that $k_i\leq k\leq k_{i+1}$. Thus,
\begin{equation}\begin{aligned}
\|\eta(k)\|=\|\Phi_{\hat{A}}(k,k_i)\eta(k_i)\|\leq \alpha_3\alpha_4^i=\alpha_3(\alpha_4^{i/k})^k,\nonumber
\end{aligned}\end{equation}
for some constants $\alpha_3\geq0$ and $0\leq\alpha_4<1$. Evidently, $k_1+(i-1)\mu\leq k\leq k_1+(i-1)\nu$. This implies that
\begin{equation}\begin{aligned}
\|\eta(k)\|\leq\alpha_2(\alpha_4^{i/(k_1+(i-1)\nu)})^k\leq\alpha_3(\alpha_4^{1/k_1})^k,\nonumber
\end{aligned}\end{equation}
where $0\leq\alpha_4^{1/k_1}<1$. The proof is completed.
\end{proof}

\begin{remark}
In the proof of Theorem \ref{TheoremLya}, the matrix $\Lambda(k)$ ensures that the stability of each subsystem remains unaffected by the global LIS, thereby achieving decoupling, as evidenced in (31) and (32). Without this matrix, the stability of each subsystem would be related to the size of the global LIS (i.e., the decoupling variable $\vartheta_i(k)=s$), not only to its neighbors. %Additionally, the proposed Lyapunov function (23) degenerates to a common form $V(k)=\zeta^T(x)\big(P^{-1}(k)+\epsilon I\big)\zeta(k)$ if $s=1$. Therefore, the analysis in Theorem 1 is a generalization of the classical Lyapunov analysis method to the time-varying LIS.
\end{remark}
\subsection{Stability of the estimate error system: A Markov analysis method}
In the last subsection, we have proven that the distributed estimator (3) is exponentially stable using the Lyapunov method. One of the key conditions is the decoupling variable $\vartheta_i(k)=\sqrt{|\mathbb{O}(k+1)|+1}$. In this subsection, we will further show that the decoupling variable $\vartheta_i(k)$ can also be chosen as $\sqrt{|\mathbb{I}(k)|+1}$. The following proposition, which is the core idea of this subsection, shows that there exists an equivalent Markov system for any LIS.

\begin{proposition}\label{PrLISeqMJP}
Consider a time-varying LIS
\begin{equation}\begin{aligned}
\zeta_{i}(k+1)=\Gamma_{ii}(k)\zeta_{i}(k)+\sum_{j\in\mathbb{I}^\Gamma_i(k)}\Gamma_{ij}(k)\zeta_{j}(k),
\end{aligned}\end{equation}
where $\Gamma_{ij}(k)\in\mathbb{R}^{n_i\times n_j}$ and
\begin{equation}
\mathbb{I}^\Gamma_i(k)\triangleq\{j|\Gamma_{ij}(k)\neq0,j=1,\cdots,s,j\neq i\}.\nonumber
\end{equation}
Meanwhile, consider a time-varying Markov system
\begin{equation}
\xi(k+1)=\big(|\mathbb{I}^\Gamma_i(k)|+1\big)\Gamma_{\varpi(k+1)\varpi(k)}(k)\xi(k),
\end{equation}
where $\{\varpi(k)\}$ is a Markov process takes the value in the set $\{1,2,\cdots,s\}$ with the transition probability
\begin{equation}\begin{aligned}
p_{ij}(k)&\triangleq\mathrm{Pr}\big(\varpi(k+1)=i|\varpi(k)=j\big)\\
&=
\left\{ \begin{array}{l}
1/\big(|\mathbb{I}^\Gamma_i(k)|+1\big),\ \mathrm{if}\ j\in\mathbb{I}^\Gamma_i(k)\cup\{i\},\\
0,\ \mathrm{if}\ j\notin\mathbb{I}^\Gamma_i(k)\cup\{i\}.
\end{array} \right.\nonumber
\end{aligned}\end{equation}
Define
\begin{equation}\begin{aligned}
\xi_i(k)\triangleq\mathrm{Pr}\big(\varpi(k)=i\big)\mathrm{E}[\zeta(k)|\varpi(k)=i],
\end{aligned}\end{equation}
and let $\xi_i(0)=\zeta_i(0)$. Then, the dynamic behavior of $\zeta_i(k)$ is equivalent to that of $\xi_i(k)$, i.e., $\zeta_i(k)=\xi_i(k)$ for all $i=1,2,\cdots,s$ and $k\geq0$.
\end{proposition}

\begin{proof}
It follows from the Bayes' theorem that
\begin{equation}\begin{aligned}
&\mathrm{Pr}\big(\varpi(k+1)=i\big)\mathrm{p}\big(\xi(k+1)|\varpi(k+1)=i\big)\\
%=&\mathrm{p}(\xi(k+1),\varpi(k+1)=i)\\
=&\sum^s_{j=1}\mathrm{p}\big(\xi(k+1),\varpi(k+1)=i,\varpi(k)=j\big)\\
%=&\sum^s_{j=1}\mathrm{p}(\xi(k+1)|\varpi(k+1)=i,\varpi(k)=j)\mathrm{p}(\varpi(k+1)=i|\varpi(k)=j)\mathrm{p}(\varpi(k)=j)\\
=&\sum_{j\in\mathbb{I}^\Gamma_i(k)\cup\{i\}}\mathrm{p}\big(\xi(k+1)|\varpi(k+1)=i,\varpi(k)=j\big)\\
&\ \ \ \ \times p_{ij}(k)\mathrm{Pr}\big(\varpi(k)=j\big).
\end{aligned}\end{equation}
Based on (36) and (37), one has
\begin{align}
\xi_i(k+1)
%=&\mathrm{p}(\varpi(k+1)=i)\mathrm{E}[\zeta(k+1)|\varpi(k+1)=i]\\
%=&\sum^s_{j=1}p_{ij}(k)\mathrm{p}(\varpi(k)=j)\mathrm{E}[\xi(k+1)|\varpi(k+1)=i,\varpi(k)=j]\\
=&\sum^s_{j=1}p_{ij}(k)\mathrm{Pr}\big(\varpi(k)=j\big)\big(|\mathbb{I}^\Gamma_i(k)|+1\big)\Gamma_{ij}(k)\nonumber\\
&\ \ \ \ \times\mathrm{E}[\xi(k)|\varpi(k+1)=i,\varpi(k)=j]\nonumber\\
%=&\sum^s_{j=1}p_{ij}(k)\mathrm{p}(\varpi(k)=j)(|\mathbb{I}(k)|+1)\Gamma_{ij}(k)\mathrm{E}[\xi(k)|\varpi(k)=j]\\
=&\Gamma_{ii}(k)\xi_i(k)+\sum_{j\in\mathbb{I}_i^\Gamma(k)}\Gamma_{ij}(k)\xi_j(k).
\end{align}
A comparison of (34) and (38) shows that $\xi_i(k)=\zeta_i(k)$ for all $i=1,2,\cdots,s$ and $k\geq0$ if $\xi_i(0)=\zeta_i(0)$ for $i=1,2,\cdots,s$. The proof is completed.
\end{proof}

Let us define
\begin{equation}\begin{aligned}
\bar{\mathfrak{F}}_k(X)\triangleq\sum^{s}_{i=1}\mathcal{A}_{i}(k)X\mathcal{A}_{i}^\top(k)+Q(k),
\end{aligned}\end{equation}
\begin{equation}\begin{aligned}
\mathfrak{F}_k(X)\triangleq X-&XC^\top(k)\big(C(k)XC^\top(k)+R(k)\big)^{-1}\\
&\times C(k)X,
\end{aligned}\end{equation}
where
\begin{equation}\begin{aligned}
\mathcal{A}_{i}(k)\triangleq\vartheta_i(k)E_iA(k),
\end{aligned}\end{equation}
\begin{equation}\begin{aligned}
E_i\triangleq\mathrm{Diag}(\underbrace{0,\cdots,0}_{i-1\ \mathrm{times}},I,\underbrace{0,\cdots,0}_{s-i\ \mathrm{times}}).\nonumber
\end{aligned}\end{equation}

\begin{proposition}\label{PrMonotonicity}
For any $k_2\geq k_1$, one has
\begin{equation}\begin{aligned}
&(\bar{\mathfrak{F}}_{k_2}\circ\mathfrak{F}_{k_2})\circ\cdots\circ(\bar{\mathfrak{F}}_{k_1}\circ\mathfrak{F}_{k_1})(X)\\
&\ \ \ \ \ \ \ \ \geq(\bar{\mathfrak{F}}_{k_2}\circ\mathfrak{F}_{k_2})\circ\cdots\circ(\bar{\mathfrak{F}}_{k_1}\circ\mathfrak{F}_{k_1})(Y)\nonumber
\end{aligned}\end{equation}
if $X\geq Y\geq0$.
\end{proposition}

\begin{proof}
According to the monotonicity of the discrete Riccati equation \cite{Sinopoli1333199}, one has $\mathfrak{F}_{k_1}(X)\geq\mathfrak{F}_{k_1}(Y)$. Then, it is trivial that $\bar{\mathfrak{F}}_{k_1}\circ\mathfrak{F}_{k_1}(X)\geq\bar{\mathfrak{F}}_{k_1}\circ\mathfrak{F}_{k_1}(Y)$. Repeating the process completes the proof of the proposition.
\end{proof}

\begin{proposition}\label{PrInitial}
If Conditions \ref{AsLUR}, \ref{AsQRB}, and \ref{AsThetaB} are fulfilled, then the following assertions are equivalent:
\begin{itemize}
\item[1)] The sequence $\{P(k)\}_{k\geq0}$ is uniformly bounded from above for any initial condition $P(0)\geq0$.
\item[2)] The sequence $\{P(k)\}_{k\geq0}$ is uniformly bounded from above for some initial condition $P(0)\geq0$.
\end{itemize}
\end{proposition}

\begin{proof}
Evidently, the first assertion is sufficient for the second one. Suppose that the second assertion is true. Then, one can know from (5) and (6) that $\{\bar{P}(k)\}_{k\geq1}$ also is uniformly bounded from above. Consider the difference equation
\begin{equation}\begin{aligned}
X(k)=(\bar{\mathfrak{F}}_{k-1}\circ\mathfrak{F}_{k-1})\circ\cdots\circ(\bar{\mathfrak{F}}_{1}\circ\mathfrak{F}_{1})\big(\bar{\mathfrak{F}}_{0}(X)\big).
\end{aligned}\end{equation}
As in the proof of Theorem \ref{TheoremLya}, there exists constants $t>0$ and $p_l>0$ such that $\bar{P}(t)\geq p_{l}I$. Thus, for any $X\geq0$, one can find a constant $\alpha\geq0$ such that
\begin{equation}\begin{aligned}
\alpha\bar{P}(t)\geq X(t).\nonumber
\end{aligned}\end{equation}
Then, it follows from Proposition \ref{PrMonotonicity} that
\begin{equation}\begin{aligned}
(\bar{\mathfrak{F}}_{k}\circ\mathfrak{F}_{k})\circ\cdots\circ(\bar{\mathfrak{F}}_{t}\circ\mathfrak{F}_{t})\big(\alpha\bar{P}(t)\big)\geq X(k+1).
\end{aligned}\end{equation}
Recall the definition of $\mathfrak{F}_k(\cdot)$, one can derive that
\begin{equation}\begin{aligned}
\mathfrak{F}_{t}\big(\alpha\bar{P}(t)\big)=&\alpha\bar{P}(t)-\alpha\bar{P}(t)C^\top(k)\big(C(k)\bar{P}(t)C^\top(k)\\
&+\frac{1}{\alpha}R(k)\big)^{-1}C(k)\bar{P}(t)\\
\leq& \alpha\mathfrak{F}_{t}\big(\bar{P}(t)\big),\nonumber
\end{aligned}\end{equation}
which immediately implies
\begin{equation}\begin{aligned}
(\bar{\mathfrak{F}}_{t}\circ\mathfrak{F}_{t})\big(\alpha\bar{P}(t)\big)\leq \alpha(\bar{\mathfrak{F}}_{t}\circ\mathfrak{F}_{t})\big(\bar{P}(t)\big).\nonumber
\end{aligned}\end{equation}
Repeating the procedure yields
\begin{equation}\begin{aligned}
&(\bar{\mathfrak{F}}_{k}\circ\mathfrak{F}_{k})\circ\cdots\circ(\bar{\mathfrak{F}}_{t}\circ\mathfrak{F}_{t})\big(\alpha\bar{P}(t)\big)\\
&\leq \alpha(\bar{\mathfrak{F}}_{k}\circ\mathfrak{F}_{k})\circ\cdots\circ (\bar{\mathfrak{F}}_{t}\circ\mathfrak{F}_{t})\big(\bar{P}(t)\big)=\alpha\bar{P}(k+1).
\end{aligned}\end{equation}
It follows from (43) and (44) that the sequence $\{X(k)\}_{k\geq1}$ is uniformly bounded from above for any $X\geq0$. Finally, based on (5) and (6), it can be verified that the DMRE (5) and (6) can be reformulated as the compact form
\begin{equation}\begin{aligned}
\bar{P}(k)=\bar{\mathfrak{F}}_k\big(P(k)\big),\ P(k+1)=\mathfrak{F}_k\big(\bar{P}(k)\big).
\end{aligned}\end{equation}
A comparison of (42) and (45) completes the proof.
\end{proof}

Based on Propositions \ref{PrLISeqMJP}, \ref{PrMonotonicity}, and \ref{PrInitial}, we can give the following theorem.

\begin{theorem}\label{TheoremMarkov}
Consider the LIS (2) and the distributed estimator (3) with gain (4). Let $w(k)=0$ and $v(k+1)=0$ for $k\geq0$. If $\vartheta_i(k)=\sqrt{|\mathbb{I}_i(k)|+1}$ for $i=1,2,\cdots,s$ and $k\geq0$ and Conditions \ref{AsLUR}, \ref{AsQRB}, and \ref{AsPB} are fulfilled, then the estimate error system is marginally stable, that is,
\begin{equation}\begin{aligned}
\lim_{k\to\infty}\|\hat{e}(k)\|<\infty.
\end{aligned}\end{equation}
\end{theorem}

\begin{proof}
The inequality (46) is equivalent to $\lim_{k\to\infty}\|\bar{e}(k)\|<\infty$. Thus, we will analyze $\bar{e}(k)$ later.
%The one-step prediction $\bar{x}_i(k+1)$ can be formulated as
%\begin{equation}\begin{aligned}
%\bar{x}_i(k+1)=\sum_{j\in\mathbb{I}_i(k)\cup\{i\}}&A_{ij}(k)\bar{x}_j(k)+A_{ij}(k)K_{j}(k)\\
%&\times\big(z_j(k)-C_j(k)\bar{x}_j(k)\big).\\
%\end{aligned}\end{equation}
Based on (2) and (3), the local one-step prediction error can be given by
\begin{equation}\begin{aligned}
&\bar{e}_i(k+1)\\
=&\sum_{j\in\mathbb{I}_i(k)\cup\{i\}}\big(A_{ij}(k)-A_{ij}(k)K_j(k)C_j(k)\big)\bar{e}_j(k)\\
&-A_{ij}(k)K_{j}(k)v_j(k)+w_i(k).\nonumber
\end{aligned}\end{equation}
Recall the Markov system (35) and define
\begin{equation}\begin{aligned}
\Xi_i(k)\triangleq\mathrm{Pr}(\varpi(k)=i)\mathrm{E}[\xi(k)\xi^\top(k)|\varpi(k)=i].\nonumber
\end{aligned}\end{equation}
Similar to the proof of Proposition \ref{PrLISeqMJP}, one has
\begin{equation}\begin{aligned}
&\Xi_i(k+1)\\
=&\sum_{j\in\mathbb{I}^\Gamma_i(k)\cup\{i\}}\big(|\mathbb{I}^\Gamma_i(k)|+1\big)\mathrm{Pr}\big(\varpi(k)=j\big)\Gamma_{ij}(k)\\
&\times\mathrm{E}[\xi(k)\xi^\top(k)|\varpi(k+1)=i,\varpi(k)=j]\\
=&\sum_{j\in\mathbb{I}^\Gamma_i(k)\cup\{i\}}\big(|\mathbb{I}^\Gamma_i(k)|+1\big)\Gamma_{ij}(k)\Xi_j(k)\Gamma_{ij}^\top(k).
\end{aligned}\end{equation}
Let the transition matrix of the LIS (34) be
\begin{equation}\begin{aligned}
\Gamma_{ij}(k)=\bar{A}_{ij}(k)\triangleq A_{ij}(k)-A_{ij}(k)K_j(k)C_j(k).\nonumber
\end{aligned}\end{equation}
Then, $\mathbb{I}^\Gamma_i(k)=\mathbb{I}_i(k)$. Since $w(k)=0$ and $v(k+1)=0$ for $k\geq0$, one can know from Proposition \ref{PrLISeqMJP} that $\bar{e}_i(k)=\xi_i(k)$ for $k\geq1$ if $\bar{e}_i(1)=\xi_i(1)$. Therefore, by choosing a special initial condition for the Markov system (35), one can obtain
\begin{equation}\begin{aligned}
&\|\Phi_{\bar{A}}(k,1)\bar{e}(1)\|^2=\|\bar{e}(k)\|^2=\sum^s_{i=1}\|\xi_i(k)\|^2\\
=&\sum^s_{i=1}\|\mathrm{Pr}\big(\varpi(k)=i\big)\mathrm{E}[\xi(k)|\varpi(k)=i]\|^2\\
\leq&\sum^s_{i=1}\mathrm{Pr}\big(\varpi(k)=i\big)^2\mathrm{E}[\|\xi(k)\|^2|\varpi(k)=i]\\
=&\sum^s_{i=1}\mathrm{Tr}\big(\Xi_i(k)\big),
\end{aligned}\end{equation}
where the inequality follows from the Jensen's inequality and
\begin{equation}\begin{aligned}
\bar{A}(k)\triangleq
\begin{bmatrix}
{\bar{A}_{11}(k)} & {\bar{A}_{12}(k)} & {\cdots} & {\bar{A}_{1s}(k)}\\
{\bar{A}_{21}(k)} & {\bar{A}_{22}(k)} & {\cdots} & {\bar{A}_{2s}(k)}\\
{\vdots} & {\vdots} & {\ddots} & {\vdots}\\
{\bar{A}_{s1}(k)} & {\bar{A}_{s2}(k)} & {\cdots} & {\bar{A}_{ss}(k)}
\end{bmatrix}.\nonumber
\end{aligned}\end{equation}
Based on (47), one has
\begin{align}
&\Xi(k+1)\triangleq\mathrm{Diag}\big(\Xi_1(k+1),\Xi_2(k+1),\cdots,\Xi_s(k+1)\big)\nonumber\\
=&\sum^{s}_{i=1}\big(\mathcal{A}_{i}(k)-\mathcal{K}_{i}(k)C(k)\big)\Xi(k)\big(\mathcal{A}_{i}(k)-\mathcal{K}_{i}(k)C(k)\big)^\top,
\end{align}
where
\begin{equation}\begin{aligned}
\mathcal{K}_i(k)\triangleq&\mathcal{A}_{i}(k)\bar{P}(k)C^\top(k)\\
&\times\big(C(k)\bar{P}(k)C^\top(k)+R(k)\big)^{-1}.
\end{aligned}\end{equation}
Additionally, note that the evolution of $\bar{P}(k)$ can be reformulated as a Lyapunov equation
\begin{align}
&\bar{P}(k+1)\nonumber\\
=&\sum^{s}_{i=1}\big(\mathcal{A}_{i}(k)-\mathcal{K}_{i}(k)C(k)\big)\bar{P}(k)\big(\mathcal{A}_{i}(k)-\mathcal{K}_{i}(k)C(k)\big)^\top\nonumber\\
&\ \ \ \ \ +\mathcal{K}_{i}(k)R(k)\mathcal{K}^\top_{i}(k)+Q(k).
\end{align}
With Condition \ref{AsPB} and Proposition \ref{PrInitial}, a comparison of (49) and (51) shows that $\{\Xi(k)\}_{k\geq1}$ is uniformly bounded from above for any initial condition $\Xi_i(0)\geq0$. This also implies that $\{\|\bar{e}(k)\|\}_{k\geq1}$ is uniformly bounded from above for any initial condition $\bar{e}(k)\in\mathbb{R}^n$. The proof is completed.
\end{proof}

\begin{remark}
Proposition \ref{PrLISeqMJP} reveals that the state trajectory of a LIS will align perfectly with the state trajectory of a Markov system. Meanwhile, we also provide an explicit construction method for this equivalent Markov system, as illustrated in (35). This conclusion holds significant implications, as it allows us to apply the stochastic theory to analyze the performance of the LIS, as demonstrated in Theorem \ref{TheoremMarkov}.
\end{remark}

\begin{remark}
Since the transition matrix $\Gamma_{ij}(k)$ of the LIS (34) is not necessarily a square matrix, the state dimension of the system (35) is time-varying. However, the column of $\Gamma_{\varpi(k+1)\varpi(k)}$ coincides with the row of $\Gamma_{\varpi(k)\varpi(k-1)}$, thus the system (57) is well-defined. Moreover, this unusual phenomenon does not affect the analysis in this subsection since we do not use the property that $\Gamma_{ij}(k)$ is a square matrix.
\end{remark}

\subsection{Boundedness of the DMRE}
This subsection establishes the relationship between the boundedness of the DMRE and the coupling strength of the LIS.

\begin{theorem}\label{TheoremCoupling}
Let Conditions \ref{AsQRB} and \ref{AsThetaB} are fulfilled. There exists a constant $c\geq0$ such that if $\|A_{ij}(k)\|\leq a$ for all $i\neq j$ and $k\geq0$, then the uniform detectability of $\big(C_i(k),\sqrt{2}\vartheta_i(k)A_{ii}(k)\big)$ implies that $\{P(k)\}_{k\geq 0}$ is uniformly bounded from above.
\end{theorem}

\begin{proof}
The equation (34) implies that $\{P(k)\}_{k\geq0}$ is uniformly bounded from above if $\{\bar{P}(k)\}_{k\geq1}$ is uniformly bounded from above. Thus, it remains to show the boundedness of $\{\bar{P}(k)\}_{k\geq1}$.

Decompose $\mathcal{A}_i(k)$ into self and coupling parts, that is,
\begin{equation}\begin{aligned}
\mathcal{A}_i(k)=\mathcal{A}^S_i(k)+\mathcal{A}^C_i(k),\nonumber
\end{aligned}\end{equation}
where
\begin{equation}\begin{aligned}
\mathcal{A}^S_i(k)\triangleq \vartheta_i(k)E_i\mathrm{Diag}\big(A_{11}(k),A_{22}(k),\cdots,A_{ss}(k)\big),\nonumber
\end{aligned}\end{equation}
\begin{equation}\begin{aligned}
\mathcal{A}^C_i(k)\triangleq \mathcal{A}_i(k)-\mathcal{A}^S_i(k).\nonumber
\end{aligned}\end{equation}
Then, construct the difference equation
\begin{equation}\begin{aligned}
X(k+1)=&2\hat{\mathcal{A}}^S(k)X(k)\big(\hat{\mathcal{A}}^S(k)\big)^\top\\
&+\mathcal{F}(k)R(k)\mathcal{F}^\top(k)+Q(k)\\
&+2\sum^{s}_{i=1}\mathcal{A}^C_i(k)X(k)\big(\mathcal{A}^C_i(k)\big)^\top\nonumber
\end{aligned}\end{equation}
with the initial condition $X(1)=\bar{P}(1)$, where
\begin{equation}\begin{aligned}
&\hat{\mathcal{A}}^S(k)\triangleq \mathcal{A}^S(k)-\mathcal{F}(k)C(k),\\
&\mathcal{A}^S(k)\triangleq\sum^s_{i=1}\mathcal{A}^S_i(k),\\
&\mathcal{F}(k)=\mathrm{Diag}\big(\mathcal{F}_1(k),\mathcal{F}_2(k),\cdots,\mathcal{F}_s(k)\big).\nonumber
\end{aligned}\end{equation}
We will prove that $X(k)\geq\bar{P}(k)$ for any block diagonal matrix $\mathfrak{G}(k)$. At the initial moment, it is obvious that $X(1)\geq\bar{P}(1)$. Suppose $X(k)\geq\bar{P}(k)$. Then, according to the optimality of the Kalman gain $\mathcal{K}_i(k)$ \cite{Optimalfiltering}, one has
\begin{equation}\begin{aligned}
&\bar{P}(k+1)\\
\leq&\sum^{s}_{i=1}\big(\mathcal{A}_{i}(k)-\mathcal{H}_{i}(k)C(k)\big)\bar{P}(k)\big(\mathcal{A}_{i}(k)-\mathcal{H}_{i}(k)C(k)\big)^\top\\
&\ \ \ \ \ +\mathcal{H}_{i}(k)R(k)\mathcal{H}_{i}^\top(k)+Q(k)\\
%\leq&\big(\mathfrak{A}^S(k)-\mathfrak{G}(k)C(k)\big)\bar{Y}(k)\big(\mathfrak{A}^S(k)-\mathfrak{G}(k)C(k)\big)^T\\
%&+\mathfrak{G}(k)(t)R(k)\mathfrak{G}^T(k)+\Theta(k)Q(k)\\
%&+\sum^{s}_{i=1}\mathfrak{A}^C_i(k)\bar{Y}(k)(\mathfrak{A}^C_i(k))^T\\
\leq&X(k+1),\nonumber
\end{aligned}\end{equation}
where $\mathcal{H}_i(k)=E_i\mathcal{F}(k)$ and the second inequality follows from Lemma 2.2 in \cite{WANG1992139}. This completes the induction. It remains to prove that, for some bounded block diagonal matrix sequence $\{\mathcal{F}(k)\}_{k\geq0}$, the sequence $\{X(k)\}_{k\geq1}$ is uniformly bounded from above.

Vectorizing the matrix $\bar{X}(k)$ yields
\begin{equation}\begin{aligned}
\mathrm{vec}[X(k+1)]=\Psi(k)\mathrm{vec}[X(k)]+\psi(k),
\end{aligned}\end{equation}
where
\begin{equation}\begin{aligned}
\Psi(k)\triangleq 2\big(\hat{\mathcal{A}}^S(k)\otimes\hat{\mathcal{A}}^S(k)+\sum^s_{i=1}\mathcal{A}^C_i(k)\otimes \mathcal{A}^C_i(k)\big),\nonumber
\end{aligned}\end{equation}
\begin{equation}\begin{aligned}
\psi(k)\triangleq \mathrm{vec}[\mathcal{F}(k)R(k)\mathcal{F}^\top(k)+Q(k)].\nonumber
\end{aligned}\end{equation}
Recursively applying (52) yields
\begin{equation}\begin{aligned}
\mathrm{vec}[X(k+N)]&=\Phi_{\Psi}(k+N,k)\mathrm{vec}[X(k)]\\
&+\sum^{N-1}_{i=0}\Phi_{\Psi}(k+N,k+i+1)\psi(k+i).\nonumber
\end{aligned}\end{equation}
Since $\big(C_i(k),\sqrt{2}\vartheta_i(k)A_{ii}(k)\big)$ is uniformly detectable for all $i=1,\cdots,s$, we know from the Kalman filtering theory \cite{Optimalfiltering} that there exists a bounded block diagonal matrix sequence $\{\mathcal{F}(k)\}_{k\geq0}$ such that the closed-loop system induced by $\sqrt{2}\hat{\mathcal{A}}^S(k)$ is exponentially stable (recall that $\hat{\mathcal{A}}^S(k)$ is block diagonal), i.e.,
\begin{equation}\begin{aligned}
\|\sqrt{2}\hat{\mathcal{A}}^S(k+N)\times\cdots\times\sqrt{2}\hat{\mathcal{A}}^S(k)\|\leq\alpha_1\alpha_2^N,\nonumber
\end{aligned}\end{equation}
where  $\alpha_1\geq0$ and $0\leq\alpha_2<1$ are constants, $k$ can be any positive integer. Then, utilizing Lemma \ref{LemmaA2} in Appendix yields
\begin{equation}\begin{aligned}
\|\big(2\hat{\mathcal{A}}^S(k+N)\otimes\hat{\mathcal{A}}^S(k+N)\big)\times\cdots&\times\big(2\hat{\mathcal{A}}^S(k)\otimes\hat{\mathcal{A}}^S(k)\big)\|\\
&\leq \sqrt{n}\alpha_1^2(\alpha_2^2)^N.
\end{aligned}\end{equation}
If $\|A_{ij}(k)\|\leq a$, according to the triangle inequality and submultiplicativity of the matrix norm, one can derive that
\begin{equation}\begin{aligned}
\|\sum^s_{i=1}\mathcal{A}^C_i(k)\otimes \mathcal{A}^C_i(k)\|\leq&\sum^s_{i=1}\|\mathcal{A}^C_i(k)\|^2\\
\leq& s(s-1)^2a^2.
\end{aligned}\end{equation}
With (53) and (54), one has
\begin{equation}\begin{aligned}
\|\Phi_{\Psi}(k+N,k)\|\leq \sqrt{n}\alpha_1^2(\alpha_2^2)^N+\mathfrak{G}(a),\nonumber
\end{aligned}\end{equation}
where $\mathfrak{G}(a)$ is a monotonically increasing function on $\mathbb{R}^1$ and $\mathfrak{G}(0)=0$. This implies that we can choose a large $N$ and a small $a$ such that $\|\Phi_{\Psi}(k+N,k)\|<\alpha_3$ for all $k$, where $0\leq\alpha_3<1$ is a constant. Meanwhile, based on Conditions \ref{AsQRB} and \ref{AsPB}, one knows that $\{\|\psi(k)\|\}_{k\geq1}$ is uniformly bounded from above. Thus, the system (39) is exponentially bounded. Consequently, the sequences $\{X(k)\}_{k\geq1}$ is uniformly bounded from above. The proof is completed. \end{proof}

\begin{corollary}\label{CoCoupling}
Let Conditions \ref{AsQRB} and \ref{AsThetaB} are fulfilled. There exists a constant $c\geq0$ such that if $\|A_{ij}(k)\|\leq a$ for all $i\neq j$ and $k\geq0$, then the uniform reachability of $\big(A_{ii}^\top(k),C_i^\top(k)\big)$ implies that $\{P(k)\}_{k\geq 0}$ is uniformly bounded from above.
\end{corollary}

\begin{proof}
This result immediately follows from the fact that the uniform reachability of $\big(A_{ii}^\top(k),C_i^\top(k)\big)$ implies the uniform detectability of $\big(C_i(k),\alpha(k)A_{ii}(k)\big)$, where $\{\alpha(k)\}_{k\geq0}$ is a real number sequence satisfying $\alpha_l\leq\|\alpha(k)\|\leq\alpha_u$ for some $0<\alpha_l\leq\alpha_u$. The proof is completed.
\end{proof}

\begin{remark}
Theorem \ref{TheoremCoupling} and Corollary \ref{CoCoupling} expose an intuitive phenomenon in the LIS, i.e., if the coupling among subsystems is weak, the detectability or reachability of the subsystem itself is sufficient to guarantee the boundedness of the DMRE as well as the stability of the distributed estimator.
\end{remark}

%\begin{remark}
%It is noteworthy that the propositions and theorems in this paper take $0$ as the initial moment for ease of presentation. In fact, by shifting the time axis, we can ascertain that these results are valid for any finite initial moment $k_0<\infty$. In such case, the parameters $a$ and $c$ in Theorem 1 may be related to the initial error $\|\hat{e}(k_0)\|$.
%\end{remark}
\section{Stability Analysis for The Time-Invariant Case}
This section will consider the time-invariant case. In this scenario, all system parameters will be time-invariant, that is,
\begin{equation}\begin{aligned}
\left\{ \begin{array}{l}
A(k)=A,\ C(k+1)=C,\\
Q(k)=Q,\ R(k+1)=R,\\
\mathbb{I}_i(k)=\mathbb{I}_i,\ \mathbb{O}_i(k)=\mathbb{O}_i,\ \vartheta_i(k)=\vartheta,
\end{array} \right.
\end{aligned}\end{equation}
for all $k\geq0$ and $i=1,2,\cdots,s$.
% Accordingly, Condition \ref{AsLUR} change to the following form.
%\begin{condition}\label{AsTILR}
%For all $i=1,\cdots,s$, $(A_{ii},\sqrt{Q_i})$ is reachable.
%\end{condition}
In the time-invariant case, by replacing the uniform detectability and reachability with detectability and reachability, Theorems \ref{TheoremLya}, \ref{TheoremMarkov}, and \ref{TheoremCoupling} still hold. We will not duplicate the proofs of these results. However, it should be emphasized that the marginal stability in Theorem \ref{TheoremMarkov} can be further improved to the exponential stability in the time-invariant case. Thus, to show this, it is necessary to discuss the boundedness and convergence of the DMRE (5) and (6).

%Clearly, the inequalities (17) and (18) in Theorem 1 will hold if Assumptions 5, 7 and 8 hold and $\vartheta_i=\sqrt{|\mathbb{O}_i|+1}$.

\begin{remark}
The system parameters must be bounded in the time-invariant case. Therefore, Conditions \ref{AsQRB} and \ref{AsThetaB} are no longer necessary in this section.
\end{remark}

\subsection{Boundedness and convergence of the DMRE}
The following theorem will derive some sufficient and necessary conditions for Condition \ref{AsPB}.

\begin{theorem}\label{TheoremTIPiif}
Consider the DMRE (5) and (6). Let (55) hold and $\vartheta\neq0$. If Condition \ref{AsLUR} is fulfilled, then the following assertions are equivalent:
\begin{itemize}
\item[1)] The sequence $\{P(k)\}_{k\geq0}$ is uniformly bounded from above for any initial condition $P(0)\geq0$.
\item[2)] There exists $\mathcal{G}\triangleq\{\mathcal{G}_1$, $\mathcal{G}_2$, $\cdots$, $\mathcal{G}_s\}$ such that $\rho(\mathfrak{L}_{\mathcal{G}})<1$, where
    \begin{equation}\begin{aligned}
    \mathfrak{L}_{\mathcal{G}}(X)\triangleq\sum^s_{i=1}(\mathcal{A}_{i}-\mathcal{G}_{i}C)X(\mathcal{A}_{i}-\mathcal{G}_{i}C)^\top
    \end{aligned}\end{equation}
    and $\mathcal{A}_i$ is defined in (36).
\item[3)] There exist $X>0$ and $Y_1$, $\cdots$, $Y_s$ such that the LMI
\begin{equation}\begin{aligned}
\begin{bmatrix}
{X} & {X\mathcal{A}_1-Y_1C} & {\cdots} & {X\mathcal{A}_s-Y_sC}\\
{\star} & {X} & {0} & {0}\\
{\star} & {\star} & {\ddots} & {0}\\
{\star} & {\star} & {\star} & {X}\\
\end{bmatrix}>0
\end{aligned}\end{equation}
is feasible.
\end{itemize}
\end{theorem}

\begin{proof}
Based on (45), one knows that the assertion 1) holds if and only if the sequence $\{\bar{P}(k)\}_{k\geq1}$ is uniformly bounded from above for any initial condition $\bar{P}(0)\geq1$. Therefore, we will subsequently analyze $\bar{P}(k)$ instead of $P(k)$.

1) $\Longrightarrow$ 2). Consider a special case of $\bar{P}(1)=0$. It follows from (45) that $\bar{P}(2)=Q\geq \bar{P}(1)$. Then, utilizing the monotonicity of the discrete Riccati equation yields $\bar{P}(3)\geq \bar{P}(2)$. By repeating the process, one has $\bar{P}(k+1)\geq \bar{P}(k)$ for $k\geq1$. As a result, the sequence $\{\bar{P}(k)\}_{k\geq1}$ is monotonically increasing and uniformly bounded from above. This implies that $\lim_{k\to\infty}\bar{P}(k)=\bar{P}$ for some $\bar{P}\geq0$ if $\bar{P}(1)=0$. Meanwhile, note that $\bar{P}(k)$ can also be represented as (51). This implies that
\begin{equation}\begin{aligned}
\bar{P}=&\sum^{s}_{i=1}(\mathcal{A}_{i}-\mathcal{K}_{i}C)\bar{P}(\mathcal{A}_{i}-\mathcal{K}_{i}C)^\top\\
&+\mathcal{K}_iR\mathcal{K}_i^\top+Q,
%\leq&\sum^N_{i=1}(\breve{\mathfrak{A}}_{i}-G_{i}C)(U(k)-\breve{U})(\breve{\mathfrak{A}}_{i}-G_{i}C)^T
\end{aligned}\end{equation}
where
\begin{equation}\begin{aligned}
\mathcal{K}_i\triangleq\mathcal{A}_{i}\bar{P}C^\top\big(C\bar{P}C^\top+R\big)^{-1}.\nonumber
\end{aligned}\end{equation}
Based on (58), some direct recursions yield
\begin{equation}\begin{aligned}
\bar{P}=&\mathfrak{L}_{\mathcal{K}}^n(\bar{P})+\Upsilon,\nonumber
\end{aligned}\end{equation}
where
\begin{equation}\begin{aligned}
\Upsilon\triangleq\sum^n_{i=1}\mathfrak{L}_{\mathcal{K}}^{i-1}(\sum^s_{j=1}U_j),\nonumber
\end{aligned}\end{equation}
\begin{equation}\begin{aligned}
\mathcal{K}\triangleq\{\mathcal{K}_1,\mathcal{K}_2,\cdots,\mathcal{K}_s\},\nonumber
\end{aligned}\end{equation}
\begin{equation}\begin{aligned}
U_i\triangleq \mathcal{K}_i R\mathcal{K}_i^\top+\mathrm{Diag}(\delta_{i1}Q_{1},\delta_{i2}Q_{2},\cdots,\delta_{is}Q_{s}).\nonumber
\end{aligned}\end{equation}
Recall the definition of $\mathcal{A}_i$ in (41), one knows that $\mathcal{K}_i$ can be formulated as a block matrix with nonzero sub-block only on the $i$th row, i.e.,
\begin{equation}\begin{aligned}
\mathcal{K}_i=
\begin{bmatrix}
{\delta_{i1}\mathcal{K}_{11}} & {\delta_{i1}\mathcal{K}_{12}} & {\cdots} & {\delta_{i1}\mathcal{K}_{1s}}\\
{\delta_{i2}\mathcal{K}_{21}} & {\delta_{i2}\mathcal{K}_{22}} & {\cdots} & {\delta_{i2}\mathcal{K}_{2s}}\\
{\vdots} & {\vdots} & {\ddots} & {\vdots}\\
{\delta_{is}\mathcal{K}_{s1}} & {\delta_{is}\mathcal{K}_{s2}} & {\cdots} & {\delta_{is}\mathcal{K}_{ss}}
\end{bmatrix},\nonumber
\end{aligned}\end{equation}
where $\mathcal{K}_{ij}\in\mathbb{R}^{n_i\times m_j}$. Then, let us define
\begin{equation}\begin{aligned}
V_i\triangleq & \mathrm{Diag}(\delta_{i1}\mathcal{K}_{11}R_{1}\mathcal{K}^\top_{11},\delta_{i2}\mathcal{K}_{22}R_{2}\mathcal{K}^\top_{22},\cdots,\delta_{is}\mathcal{K}_{ss}R_{s}\mathcal{K}^\top_{ss})\\
&+\mathrm{Diag}(\delta_{i1}Q_{1},\delta_{i2}Q_{2},\cdots,\delta_{is}Q_{s}),\nonumber
\end{aligned}\end{equation}
\begin{equation}\begin{aligned}
\mathfrak{L}_i(X)\triangleq(\mathcal{A}_{i}-\mathcal{K}_{i}C)X(\mathcal{A}_{i}-\mathcal{K}_{i}C)^\top.\nonumber
\end{aligned}\end{equation}
Evidently, $U_i\geq V_i$ for all $i=1,2,\cdots,s$. In this case, it can be derived that
\begin{equation}\begin{aligned}
\sum^n_{i=1}\mathfrak{L}_{\mathcal{K}}^{i-1}(\sum^s_{j=1}U_j)&\geq\sum^n_{i=1}\sum^s_{j=1}\mathfrak{L}^{i-1}_j(V_j)\\
&\geq\mathrm{Diag}(\mathcal{R}_1,\mathcal{R}_2,\cdots,\mathcal{R}_s),\nonumber
\end{aligned}\end{equation}
where
\begin{equation}\begin{aligned}
\mathcal{R}_i=\sum^{n_i}_{j=1}&(A_{ii}-\mathcal{K}_{ii}C_i)^{j-1}(Q_{i}+\mathcal{K}_{ii}R_{i}\mathcal{K}^\top_{ii})\\
&\times\big((A_{ii}-\mathcal{K}_{ii}C_i)^{j-1}\big)^\top.\nonumber
\end{aligned}\end{equation}
Note that $\mathcal{R}_i$ is the reachability Gramian of the matrix pair $(A_{ii}-\mathcal{K}_{ii}C_i,[\sqrt{Q_{i}},\ \mathcal{K}_{ii}\sqrt{R_i}])$.
% and $A_{ii}-\mathcal{K}_{ii}C_i$ and $Q_{i}+\mathcal{K}_{ii}R_{i}\mathcal{K}_{ii}^\top$ can be reformulated as
%\begin{equation}\begin{aligned}
%A_{ii}-\mathcal{K}_{ii}C_i=A_{ii}-[\sqrt{Q_{i}},\ \mathcal{K}_{ii}\sqrt{R_i}][0;\ \sqrt{R_i}^{-1}C_i],\nonumber
%\end{aligned}\end{equation}
%\begin{equation}\begin{aligned}
%Q_{i}+\mathcal{K}_{ii}R_{i}\mathcal{K}^\top_{ii}=[\sqrt{Q_{i}},\ \mathcal{K}_{ii}\sqrt{R_i}][\sqrt{Q_{i}},\ \mathcal{K}_{ii}\sqrt{R_i}]^\top.\nonumber
%\end{aligned}\end{equation}
At the same time, under Condition \ref{AsLUR}, one knows from the invariance of reachability under feedback that $\mathcal{R}_i>0$ for all $i=1,\cdots,s$. This implies that $\bar{P}>0$ and $\Upsilon>0$.

Consider the Lyapunov function
\begin{equation}\begin{aligned}
\varphi(X)=\langle X,\bar{P}\rangle.
\end{aligned}\end{equation} Since $\bar{P}>0$, it is trivial that the Lyapunov function (59) satisfies the following three conditions in $\mathbb{S}^n_{\geq0}$:
\begin{itemize}
\item[i)] $\varphi(X)$ is continuous,
\item[ii)] $\varphi(X)\to\infty$ if $\|X\|\to\infty$,
\item[iii)] $\varphi(X)=0$ if $X=0$ and $\varphi(X)>0$ if $X\neq 0$,
\end{itemize}
where $\mathbb{S}^n_{\geq0}$ is the symmetric positive semi-definite matrix set in $\mathbb{R}^{n\times n}$. Further consider the system
\begin{equation}\begin{aligned}
X(k+1)=(\mathfrak{L}^*_{\mathcal{K}})^n\big(X(k)\big),
\end{aligned}\end{equation}
where $\mathfrak{L}^*_{\mathcal{K}}(X)\triangleq\sum^s_{i=1}(\mathcal{A}_{i}-\mathcal{K}_{i}C)^\top X(\mathcal{A}_{i}-\mathcal{K}_{i}C)$ is the adjoint operator of $\mathfrak{L}_{\mathcal{K}}$ (see Lemma \ref{LemmaA3}). Then, one can derive that
\begin{equation}\begin{aligned}
&\varphi(X(k+1))-\varphi(X(k))\\
=&\langle\bar{P},(\mathfrak{L}^*_{\mathcal{K}})^n\big(X(k)\big)\rangle-\langle\bar{P},X(k)\rangle\\
=&\langle \mathfrak{L}_{\mathcal{K}}^n(\bar{P})-\bar{P},X(k)\rangle\\
=&-\langle\Upsilon,X(k)\rangle<0,\nonumber
\end{aligned}\end{equation}
if $X(k)\neq0$. Thus, one can know from the Lyapunov theory \cite{vidyasagar2002nonlinear} that the system (60) is asymptotically stable. This implies that
\begin{equation}\begin{aligned}
\rho(\mathfrak{L}_{\mathcal{K}})^n=\rho(\mathfrak{L}^*_{\mathcal{K}})^n=\rho\big((\mathfrak{L}^*_{\mathcal{K}})^n\big)<1.\nonumber
\end{aligned}\end{equation}

2) $\Longrightarrow$ 3). Since $\rho(\mathfrak{L}_{\mathcal{G}})<1$, there exists $X>0$ such that
\begin{equation}\begin{aligned}
X>\mathfrak{L}_{\mathcal{G}}(X)=\sum^{s}_{i=1}(\mathcal{A}_{i}-\mathcal{G}_{i}C)X(\mathcal{A}_{i}-\mathcal{G}_{i}C)^\top.
\end{aligned}\end{equation}
According to Schur complement lemma \cite{boyd1994linear}, one knows that (61) holds if and only if
\begin{equation}\begin{aligned}
\begin{bmatrix}
{X} & {\mathcal{A}_{1}-\mathcal{G}_{1}C} & {\cdots} & {{\mathcal{A}_{s}-\mathcal{G}_{s}C}}\\
{\star} & {X^{-1}} & {0} & {0}\\
{\star} & {\star} & {\ddots} & {0}\\
{\star} & {\star} & {\star} & {X^{-1}}\\
\end{bmatrix}>0.
\end{aligned}\end{equation}
Then, multiplying $\mathrm{Diag}(X^{-1},I,\cdots,I)$ on both sides of (62) yields (57).

3) $\Longrightarrow$ 1). Similar to the proof above, it follows from Schur complement lemma that the LMI (57) is feasible if and only if there exists $X>0$ and $\mathcal{G}_1$, $\mathcal{G}_2$, $\cdots$, $\mathcal{G}_s$ such that
\begin{equation}\begin{aligned}
X>\sum^{s}_{i=1}(\mathcal{A}_{i}-\mathcal{G}_{i}C)X(\mathcal{A}_{i}-\mathcal{G}_{i}C)^\top.
\end{aligned}\end{equation}
Then, it follows from Proposition A.3 in \cite{Fan992124} that the solution to following difference equation is uniformly bounded from above:
\begin{equation}\begin{aligned}
X(k+1)=&\sum^{s}_{i=1}(\mathcal{A}_{i}-\mathcal{G}_{i}C)X(k)(\mathcal{A}_{i}-\mathcal{G}_{i}C)^\top\\
&\ \ \ \ +\mathcal{G}_{i}R\mathcal{G}^\top_{i}+Q,\nonumber
\end{aligned}\end{equation}
for any initial condition $X(1)\geq0$. For any $\bar{P}(1)\geq0$, let $X(1)=\bar{P}(1)$. It is obvious that $\bar{P}(1)\leq X(1)$. Suppose that $\bar{P}(k)\leq X(k)$. Then, one can derive from the optimality of the Kalman gain $\mathcal{K}_i(k)$ that
\begin{equation}\begin{aligned}
X(k+1)\geq&\sum^{s}_{i=1}(\mathcal{A}_{i}-\mathcal{G}_{i}C)\bar{P}(k)(\mathcal{A}_{i}-\mathcal{G}_{i}C)^\top\\
&\ \ \ \ +\mathcal{G}_{i}R\mathcal{G}^\top_{i}+Q\geq\bar{P}(k+1).\nonumber
\end{aligned}\end{equation}
This induction means that, for any $\mathcal{G}_1$, $\cdots$, $\mathcal{G}_s$, we have $X(k)\geq \bar{P}(k),\ \forall\ k\geq 1$. Thus, $\{\bar{P}(k)\}_{k\geq1}$ is uniformly bounded from above for any initial condition $\bar{P}(1)\geq0$. The proof is completed.
\end{proof}

Notice that when Condition \ref{AsLUR} is not fulfilled, the assertions 2) and 3) will be sufficient but not necessary for the assertion 1) in Theorem \ref{TheoremTIPiif}. Moreover, the LMI (57) provided by Theorem \ref{TheoremTIPiif} is centralized since it depends on the global parameters of the LIS. Next, a distributed stability condition will be given by searching for a feasible solution to (57).

\begin{proposition}
Define $\{\tau^i_{1},\tau^i_{2},\cdots,\tau^i_{\varsigma_i}\}\triangleq\mathbb{I}_i\cup \{i\}$, where $\varsigma_i\triangleq|\mathbb{I}_i|+1$. There exist $X>0$ and $Y_1$, $Y_2$, $\cdots$, $Y_s$ such that the LMI (57) is feasible if the LMI $\mathcal{P}_i$ is feasible with respect to $X_{i1}$, $X_{i2}$, $\cdots$, $X_{i\varsigma_i}$ for all $i=1,2,\cdots,s$, where
\begin{equation}\begin{aligned}
\begin{bmatrix}
{\frac{1}{\vartheta_i}I} & {A_{i\tau^i_{1}}-X_{i1}C_{\tau^i_{1}}} & {\cdots} & {A_{i\tau^i_{\varsigma_i}}-X_{i\varsigma_i}C_{\tau^i_{\varsigma_i}}}\\
{\star} & {\frac{1}{\vartheta_i}I} & {0} & {0}\\
{\star} & {\star} & {\ddots} & {0}\\
{\star} & {\star} & {\star} & {\frac{1}{\vartheta_i}I}
\end{bmatrix}>0.
\end{aligned}\end{equation}
%\begin{equation}\begin{aligned}
%\underbrace{\begin{bmatrix}
%{I} & {\vartheta_iA_{i\tau^i_{1}}-X_{i1}C_{\tau^i_{1}}} & {\cdots} & {\vartheta_iA_{i\tau^i_{\varsigma_i}}-X_{i\varsigma_i}C_{\tau^i_{\varsigma_i}}}\\
%{\star} & {I} & {0} & {0}\\
%{\star} & {\star} & {\ddots} & {0}\\
%{\star} & {\star} & {\star} & {I}
%\end{bmatrix}>0.}_{\mathcal{P}_i}\nonumber
%\end{aligned}\end{equation}
\end{proposition}

\begin{proof}
Recall the proof of Theorem \ref{TheoremTIPiif}, we know that (57) is feasible if and only if there exist $X>0$ and $\mathcal{G}_1$, $\cdots$, $\mathcal{G}_s$ such that (63) is feasible. Then, let the variables $X$ and $\mathcal{G}_i$ in (63) be $X=I$ and
\begin{equation}\begin{aligned}
\mathcal{G}_i=
\begin{bmatrix}
{\mathcal{G}_{i,11}} & {\mathcal{G}_{i,12}} & {\cdots} & {\mathcal{G}_{i,1s}}\\
{\mathcal{G}_{i,21}} & {\mathcal{G}_{i,22}} & {\cdots} & {\mathcal{G}_{i,2s}}\\
{\vdots} & {\vdots} & {\ddots} & {\vdots}\\
{\mathcal{G}_{i,s1}} & {\mathcal{G}_{i,s2}} & {\cdots} & {\mathcal{G}_{i,ss}}
\end{bmatrix},\nonumber
\end{aligned}\end{equation}
where $\mathcal{G}_{i,\kappa\iota}\in\mathbb{R}^{n_\kappa\times n_\iota}$ and $\mathcal{G}_{i,\kappa\iota}=0$ if $\kappa\neq i$. With these settings, one can derive that
\begin{equation}\begin{aligned}
&(\mathcal{A}_{i}-\mathcal{G}_{i}C)X(\mathcal{A}_{i}-\mathcal{G}_{i}C)^\top\\
=&\mathrm{Diag}\big(0,\sum_{j\in\mathbb{I}_i\cup\{i\}}(\vartheta_iA_{ij}-\mathcal{G}_{i,ij}C_j)(\vartheta_iA_{ij}-\mathcal{G}_{i,ij}C_j)^\top,0\big).\nonumber
\end{aligned}\end{equation}
In this case, we can know that (57) holds if there exists $X_{ij}$ such that
\begin{equation}\begin{aligned}
\sum_{j\in\mathbb{I}_i\cup\{i\}}(\vartheta_iA_{ij}-X_{ij}C_j)(\vartheta_iA_{ij}-X_{ij}C_j)^\top<I
\end{aligned}\end{equation}
for all $i=1,\cdots,s$. Finally, utilizing Schur complement lemma \cite{boyd1994linear} for (65) yields (64). The proof is completed.
\end{proof}
%{\em Remark 4.} The matrix pair $(C_{i}(t),\gamma_i(t)A_{ii}(t))$ being uniformly detectable is necessary for the matrix pair $(C_i(t),A_{ii}(t))$ to be  uniformly observable, and thus the local estimator also is stable if the relevant subsystem is uniformly observable. Moreover, Theorem 2 proves for the first time that a fundamental property of the LIS itself, namely uniform detectability, is a sufficient condition for the existence of a stable fully distributed estimator. In addition, the existing works usually require some additional conditions, such as special topology\cite{CHEN2019228} and invertible state transition matrix \cite{Farina7762729}.

%Moreover, the stability results in existing works are less general than Theorem 2 in this paper since they tend to require that the LIS be time-invariant or that the state transition matrix of the LIS be invertible, which are not required by Theorem 2.

For an estimator acting on the time-invariant LIS, besides its stability, we are also interested in the convergence of its parameters. The following theorem will prove that the distributed estimator (3) converges to a unique steady state for any initial condition.

\begin{theorem}\label{TheoremTIConvergence}
Consider the LIS (2) and the distributed estimator (3) with the local gain (4). Let (55) hold and $\vartheta\neq0$. If Condition \ref{AsLUR} is fulfilled and the LMI (57) is feasible, then the distributed estimator (3) converges to a unique steady state for any initial condition, i.e.,
\begin{equation}\begin{aligned}
\lim_{k\rightarrow\infty}\bar{P}(k)=\bar{P},
\end{aligned}\end{equation}
\begin{equation}\begin{aligned}
\lim_{k\rightarrow\infty}K(k)=K\triangleq\bar{P}C^\top(C\bar{P}C^\top+R)^{-1},
\end{aligned}\end{equation}
where $\bar{P}$ is the unique symmetric positive definite solution to the DMRE
\begin{equation}\begin{aligned}
X=\Omega\odot\big(&\mathcal{A}X\mathcal{A}^\top-\mathcal{A}XC^\top(CXC^\top+R)^{-1}\\
&\times CX\mathcal{A}^\top+Q\big),
\end{aligned}\end{equation}
where
\begin{equation}\begin{aligned}
\Omega\triangleq\mathrm{Diag}(\mathbf{1}_{n_1n_1},\mathbf{1}_{n_2n_2},\cdots,\mathbf{1}_{n_sn_s}),\nonumber
\end{aligned}\end{equation}
\begin{equation}\begin{aligned}
\mathcal{A}\triangleq
\begin{bmatrix}
{\vartheta_1A_{11}} & {\vartheta_1A_{12}} & {\cdots} & {\vartheta_1A_{1s}}\\
{\vartheta_2A_{21}} & {\vartheta_2A_{22}} & {\cdots} & {\vartheta_2A_{2s}}\\
{\vdots} & {\vdots} & {\ddots} & {\vdots}\\
{\vartheta_sA_{s1}} & {\vartheta_sA_{s2}} & {\cdots} & {\vartheta_sA_{ss}}
\end{bmatrix}.\nonumber
\end{aligned}\end{equation}
Meanwhile, if $\vartheta_i=\sqrt{|\mathbb{O}_i|+1}$, then the steady-state closed-loop matrix $A-KCA$ is stable, i.e.,
\begin{equation}\begin{aligned}
\rho(A-KCA)<1.
\end{aligned}\end{equation}
\end{theorem}

\begin{proof}
We first prove the convergence of $\bar{P}(k)$. Substituting (5) into (6) yields
\begin{align}
\bar{P}_i(k+1)=&\sum_{j\in\mathbb{I}_i\cup \{i\}}A_{ij}\bar{P}_j(k)A_{ij}^\top-A_{ij}\bar{P}_j(k)C_j^\top\nonumber\\
&\ \ \ \ \ \ \ \ \ \ \times\big(C_j\bar{P}_j(k)C_j^\top+R_{j}\big)^{-1}C_j\bar{P}_j(k)A_{ij}^\top\nonumber\\
&+Q_{i}.
\end{align}
Based on (70), the DMRE can be reformulated as
\begin{equation}\begin{aligned}
\bar{P}(k+1)=&\Omega\odot\big(\mathcal{A}\bar{P}(k)\mathcal{A}^\top+Q-\mathcal{A}\bar{P}(k)C^\top\\
&\times\big(C\bar{P}(k)C^\top+R\big)^{-1}C\bar{P}(k)\mathcal{A}^\top\big).
\end{aligned}\end{equation}

If $\bar{P}(1)=0$, as in the proof of Theorem \ref{TheoremTIPiif} that
\begin{equation}\begin{aligned}
\lim_{k\rightarrow\infty}\bar{P}(k)=\bar{P},\ \mathrm{if}\ \bar{P}(1)=0.
\end{aligned}\end{equation}

If $\bar{P}(1)\geq \bar{P}$, it follows from (71) and the monotonicity of the discrete Riccati equation that
\begin{equation}\begin{aligned}
\bar{P}(2)\geq&\Omega\odot\big(\mathcal{A}\bar{P}\mathcal{A}^\top+Q-\mathcal{A}\bar{P}C^\top\\
&\times(C\bar{P}C^\top+R)^{-1}C\bar{P}\mathcal{A}^\top\big)=\bar{P}.\nonumber
\end{aligned}\end{equation}
Repeating the procedure above gives
\begin{equation}\begin{aligned}
\bar{P}(k)-\bar{P}\geq 0,\ \forall k\geq1.
\end{aligned}\end{equation}
It follows from (45) and the optimality of the Kalman gain $\mathcal{K}_i(k)$ that
\begin{equation}\begin{aligned}
\bar{P}(k+1)\leq&\sum^{s}_{i=1}(\mathcal{A}_{i}-\mathcal{K}_{i}C)\bar{P}(k)(\mathcal{A}_{i}-\mathcal{K}_{i}C)^\top\\
&+\mathcal{K}_{i}R\mathcal{K}^\top_{i}+Q,
\end{aligned}\end{equation}
where
\begin{equation}\begin{aligned}
\mathcal{K}_{i}\triangleq \mathcal{A}_i\bar{P}C^\top(C\bar{P}C^\top+R)^{-1}.\nonumber
\end{aligned}\end{equation}
%Note that $\bar{P}$ is a fixed point of the difference equation (56). Meanwhile, (56) and (46) are equivalent. Thus, one has
%\begin{equation}\begin{aligned}
%\bar{P}=&\sum^{s}_{i=1}(\mathcal{A}_{i}-\mathcal{K}_{i}C)\bar{P}(\mathcal{A}_{i}-\mathcal{K}_{i}C)^T\\
%&+\mathcal{K}_iR\mathcal{K}_i^T+Q.
%%\leq&\sum^N_{i=1}(\breve{\mathfrak{A}}_{i}-G_{i}C)(U(k)-\breve{U})(\breve{\mathfrak{A}}_{i}-G_{i}C)^T
%\end{aligned}\end{equation}
Based on (58) and (74), one has
\begin{equation}\begin{aligned}
&\bar{P}(k)-\bar{P}\leq \mathfrak{L}_{\mathcal{K}}^{k-1}\big(\bar{P}(1)-\bar{P}\big),
\end{aligned}\end{equation}
In the proof of Theorem \ref{TheoremTIPiif}, we have prove that $\rho(\mathfrak{L}_{\mathcal{K}})<1$. As a result, one can obtain from (73) and (75) that
\begin{equation}\begin{aligned}
\lim_{t\rightarrow\infty}\bar{P}(k)=\bar{P},\ \forall\ \bar{P}(1)\geq \bar{P}.
\end{aligned}\end{equation}

For arbitrary $\bar{P}(1)\geq0$, one has $0\leq \bar{P}(1)\leq \bar{P}+\bar{P}(1)$. In this case, based on (73) and (76), one can derive from Proposition \ref{PrMonotonicity} that
\begin{equation}\begin{aligned}
\lim_{t\rightarrow\infty}\bar{P}(k)=\bar{P},\ \forall\ \bar{P}(1)\geq 0.
\end{aligned}\end{equation}
Consequently, (66) and (67) can be obtained. Moreover, since the limit (77) is independent of the initial condition, the symmetric positive definite solution to (68) is unique.

Next, the inequality (69) will be proved. Consider a special case where $P(0)=\bar{P}-\bar{P}C^\top(C\bar{P}C^\top+R)^{-1}C\bar{P}$. This implies that $\bar{P}(k)=\bar{P}$ and $K(k)=K$ for all $k\geq1$. In this case, it follows from (21) that the global estimate error can be expressed as
\begin{equation}\begin{aligned}
\hat{e}(k)=&(A-KCA)\hat{e}(k-1)+(I-KC)w(k-1)\\
&-Kv(k).
\end{aligned}\end{equation}
Utilizing the exponential convergence of the noise-free version of the system (78) (c.f. Theorem \ref{TheoremLya}) yields
\begin{equation}\begin{aligned}
\lim_{k\to\infty}\|(A-KCA)^k\|=0.\nonumber
\end{aligned}\end{equation}
Finally, according to the spectral property of bounded linear operators \cite{costa2006discrete}, one knows that $\rho(A-KCA)<1$. The proof is completed.
\end{proof}

\begin{remark}
Utilizing the feasibility of the LMI as a stability criterion is common in the control community \cite{boyd1994linear}. Such optimization problems are convex, and there are numerous well-established toolboxes, including the “feasp” function in MATLAB, for verifying the feasibility of a given LMI.
\end{remark}

\begin{remark}
Since the closed-loop matrix $A-KCA$ is stable, we can implement the following steady-state distributed estimator
\begin{equation}\begin{aligned}
\hat{x}_{st,i}(k)=&\bar{x}_{st,i}(k)+K_{i}\big(z_i(k)-C_{i}\bar{x}_{st,i}(k)\big),
\end{aligned}\end{equation}
where $\bar{x}_{st,i}(k)=A_{ii}\hat{x}_{st,i}(k-1)+\sum_{j\in\mathbb{I}_i}A_{ij}\hat{x}_{st,j}(k-1)$, $\mathrm{Diag}(K_{1},\cdots,K_{s})=K$ and $K_{i}\in\mathbb{R}^{n_i\times m_i}$. By recursing the DMRE (5) and (6) within each subsystem, we can obtain the steady-state gain $K$ off-line. Thus, compared with the original distributed estimator (3), the steady-state version (79) effectively reduces the communication burden among the subsystems during the real-time estimation.
\end{remark}

\subsection{Stability of the estimate error system}
\begin{theorem}\label{TheoremTIMarkov}
Consider the LIS (2) and the distributed estimator (3) with gain (4). Let (55) hold and $\vartheta_i=\sqrt{|\mathbb{I}_i|+1}$ for $i=1,2,\cdots,s$. If Conditions \ref{AsLUR} and \ref{AsPB} are fulfilled, then the estimate error system is exponentially stable, that is,
\begin{equation}\begin{aligned}
\|\hat{e}(k)\|\leq ab^k+c,
\end{aligned}\end{equation}
where $a\geq0$, $0\leq b<1$ and $c\geq0$ are constants. Meanwhile, if $w(k)=0$ and $v(k+1)=0$ for all $k\geq0$, then
\begin{equation}\begin{aligned}
\|\hat{e}(k)\|\leq ab^k.
\end{aligned}\end{equation}
\end{theorem}
Additionally, (80) and (81) also hold when the condition $\vartheta_i=\sqrt{|\mathbb{I}_i|+1}$ is replaced by $\vartheta_i=\sqrt{|\mathbb{O}_i|+1}$.

\begin{proof}
According to Theorems \ref{TheoremTIPiif} and \ref{TheoremTIConvergence}, one knows that $\{\bar{P}(k)\}_{k\geq1}$ is a convergent sequence with limit $\bar{P}$. Thus, consider the LIS (34) with the transition matrix
\begin{equation}\begin{aligned}
\Gamma_{ij}(k)=A_{ij}-A_{ij}K_jC_j,\ \forall k\geq1,\nonumber
\end{aligned}\end{equation}
where $K_i=\lim_{k\to\infty}K_i(k)$ is the steady-state local gain. Then, similar to the proof of (48), one has
\begin{equation}\begin{aligned}
\|(A-AKC)^{k}\zeta(0)\|\leq\mathrm{Tr}(\Xi(k)),\nonumber
\end{aligned}\end{equation}
where
\begin{equation}\begin{aligned}
\Xi(k)&=\sum^{s}_{i=1}(\mathcal{A}_{i}-\mathcal{K}_{i}C)\Xi(k-1)(\mathcal{A}_{i}-\mathcal{K}_{i}C)^\top\\
&=\mathfrak{L}_{\mathcal{K}}\big(\Xi(k-1)\big).\nonumber
\end{aligned}\end{equation}
Moreover, we have proved in Theorem \ref{TheoremTIPiif} that $\rho(\mathfrak{L}_{\mathcal{K}})<1$ if Condition \ref{AsLUR} and \ref{AsPB} are fulfilled. Thus, $\lim_{k\to\infty}\Xi(k)=0$ for any initial condition. This also implies that $\lim_{k\to\infty}\|(A-AKC)^{k}\zeta(0)\|=0$ for any $\zeta(0)$. Thus, $\rho(A-AKC)=\rho(A-KCA)<1$. In this case, there exists $X^{-1}>I$ such that
\begin{equation}\begin{aligned}
(A-KCA)X^{-1}(A-KCA)^\top=X^{-1}-I.
\end{aligned}\end{equation}
The identity (82) implies that $\|X^{1/2}(A-KCA)X^{-1/2}\|<1$. According to the definition of convergence, one knows that there exist constants $0\leq\alpha<1$ and $N>0$ such that
\begin{equation}\begin{aligned}
\|X^{1/2}(A-K(k)CA)X^{-1/2}\|\leq\alpha,\ \forall\ k\geq N.
\end{aligned}\end{equation}
Then, it can be derived from (23) and (83) that
\begin{equation}\begin{aligned}
X^{1/2}&\hat{e}(k)=X^{1/2}\big(I-K(k)C\big)AX^{1/2}X^{-1/2}\hat{e}(k-1)\\
&+X^{1/2}\big(I-K(k)C\big)w(k-1)-X^{1/2}K(k)v(k),\nonumber
\end{aligned}\end{equation}
which is exponentially stable. Finally, (80) and (81) can be obtained from the fact that $X^{1/2}$ is invertible.

Additionally, it follows from Theorem \ref{TheoremTIConvergence} that $\rho(A-KCA)<1$ if $\vartheta_i=\sqrt{|\mathbb{O}_i|+1}$. Similar to the proof above, we can conclude that (80) and (81) also hold if $\vartheta_i=\sqrt{|\mathbb{O}_i|+1}$. The proof is completed.
\end{proof}

\section{Simulation}
In this section, a multi-area power system \cite{riverso2012hycon2} is used to validate the effectiveness of the proposed distributed estimator. The dynamical process of the power system is given by
\begin{equation}\begin{aligned}
\dot{x}_{i}=A^c_{ii}(t)x_{i}+\sum_{j\in\mathbb{I}_{i}}A^c_{ij}(t)x_{j}+B^c_i(t)u_{i}+w^c_i,\nonumber
\end{aligned}\end{equation}
where $x_i=[\Delta\theta_i;\Delta\omega_i;\Delta P_{m,i};\Delta P_{\upsilon,i}]$ is the subsystem state, $u_i=[\Delta P_{ref,i};\Delta P_{L,i}]$ is the subsystem input and $w_i$ is the process noises. The system matrices is given by
\begin{equation}\begin{aligned}
A^c_{ii}=\begin{bmatrix}
{0} & {1} & {0} & {0}\\
{-\frac{\sum_{j\in\mathbb{I}_i}P_{ij}}{2H_i}} & {-\frac{D_i}{2H_i}} & {\frac{1}{2H_i}} & {0}\\
{0} & {0} & {-\frac{1}{T_{t,i}}} & {\frac{1}{T_{t,i}}}\\
{0} & {-\frac{1}{R_iT_{g,i}}} & {0} & {-\frac{1}{T_{g,i}}}
\end{bmatrix},\nonumber
\end{aligned}\end{equation}
\begin{equation}\begin{aligned}
A^c_{ij}=\begin{bmatrix}
{0} & {0} & {0} & {0}\\
{\frac{P_{ij}}{2H_i}} & {0} & {0} & {0}\\
{0} & {0} & {0} & {0}\\
{0} & {0} & {0} & {0}
\end{bmatrix}\ (i\neq j),\ B^c_{i}=\begin{bmatrix}
{0} & {0}\\
{0} & {-\frac{1}{2H_i}}\\
{0} & {0}\\
{\frac{1}{T_{g,i}}} & {0}
\end{bmatrix},\nonumber
\end{aligned}\end{equation}
where the time index of the parameters are omitted for brevity. The definitions of the model parameters are as follows
\begin{itemize}
\item $H_i$: inertial constant;
\item $R_i$: speed regulation;
\item $T_{t,i}$: prime mover time constant;
\item $T_{g,i}$: governor time constant;
\item $P_{ij}$: the slope of the power angle curve at the initial operating angle between subsystem $i$ and subsystem $j$;
\item $\Delta\theta_i$: deviation of the angular displacement of the rotor with respect to the stationary reference axis of the stator;
\item $\Delta\omega_i$: speed deviation of the rotating mass from nominal value;
\item $\Delta P_{m,i}$: deviation of the mechanical power from the nominal value;
\item $\Delta P_{\omega,i}$: deviation of the steam valve position from nominal value;
\item $\Delta P_{ref,i}$: deviation of the reference set power from nominal value;
\item $\Delta P_{L,i}$: deviation of the non-frequency-sensitive load change from nominal value.
\end{itemize}
By utilizing the discretization method in \cite{FARINA20133411}, one can obtain the discrete-time representation of the power system, that is,
\begin{equation}\begin{aligned}
x_{i}(k+1)=&A_{ii}(k)x_i(k)+\sum_{j\in\mathbb{I}_{i}}A_{ij}(k)x_{j}(k)\\
&+B_i(k)u_{i}(k)+w_i(k).\nonumber
\end{aligned}\end{equation}
Moreover, the measurement equation is given by
\begin{equation}\begin{aligned}
z_i(k)=\begin{bmatrix}
{1} & {0} & {0} & {0}\\
{0} & {1} & {0} & {0}
\end{bmatrix}x_i(k)+v_i(k),\nonumber
\end{aligned}\end{equation}
where $v_i(k)$ is the measurement noise. In this simulation, the process noise and the measurement noise are set to be uniformly distributed with covariance matrix $I$. The number of subsystem is $s=10$, the sampling period is $1s$, the model parameters change every 100 moments and are selected according to the range provided in \cite{riverso2012hycon2}.

To visualize the estimation performance, we define the root mean-square error (RMSE) of an estimator to be
\begin{equation}\begin{aligned}
\text{RMSE}(k)\triangleq\sqrt{\frac{1}{sM}\sum^s_{i=1}\sum^{M}_{j=1}\|x_{i,j}(k)-\hat{x}_{i,j}(k)\|^2},\nonumber
\end{aligned}\end{equation}
where $x_{i,j}(k)$ and $\hat{x}_{i,j}(k)$ mean the real and estimated states of $i$th subsystem at $j$th Monte Carlo trial, respectively. We choose $M=500$ in this paper. Fig. 1 illustrates the RMSE of various estimation methods. We can see from this figure that the centralized estimator performs slightly better than the proposed distributed estimator. However, in a practical LIS, the implementing centralized methods is often infeasible.
Consequently, despite a minor reduction in estimation accuracy, the proposed distributed estimator offers significant advantages in terms of practicality. Meanwhile, Fig. 1 also shows that the proposed distributed estimator outperforms the distributed estimator in \cite{Riverso6760656}, further demonstrating the advantage of the proposed method. Additionally, Fig. 2 presents the estimation performance of the proposed distributed estimator in the absence of noises. These figures indicate that the RMSE of the distributed estimator approaches 0 after 10 moments. Particularly, Theorem 3 claims that the distributed estimator is marginally stable when $\vartheta_i(k)=\mathbb{I}_i(k)$, however, the simulations reveal that asymptotic stability is also attainable.
\begin{figure}
      \centering
      \includegraphics[scale=0.6]{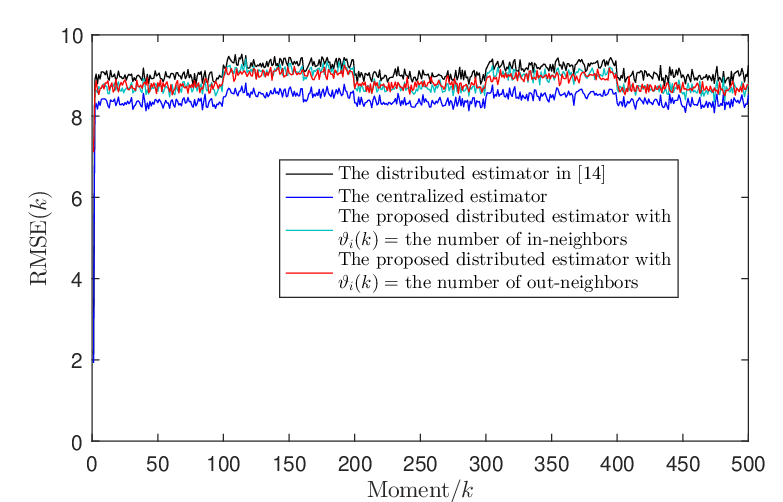}
      \caption{The RMSE of different methods in the time-varying LIS with noises.}
\end{figure}
\begin{figure}
      \centering
      \includegraphics[scale=0.6]{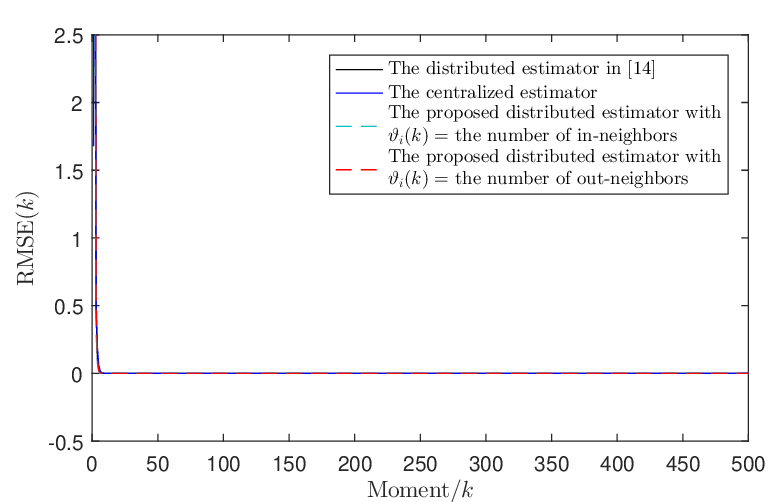}
      \caption{The RMSE of different methods in the time-varying LIS without noises.}
\end{figure}
\begin{figure}
      \centering
      \includegraphics[scale=0.6]{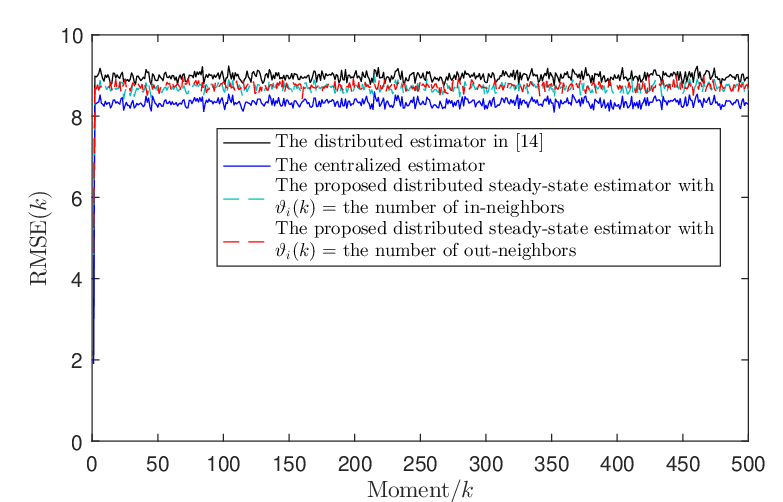}
      \caption{The RMSE of different methods in the time-invariant LIS with noises.}
\end{figure}
\begin{figure}
      \centering
      \includegraphics[scale=0.6]{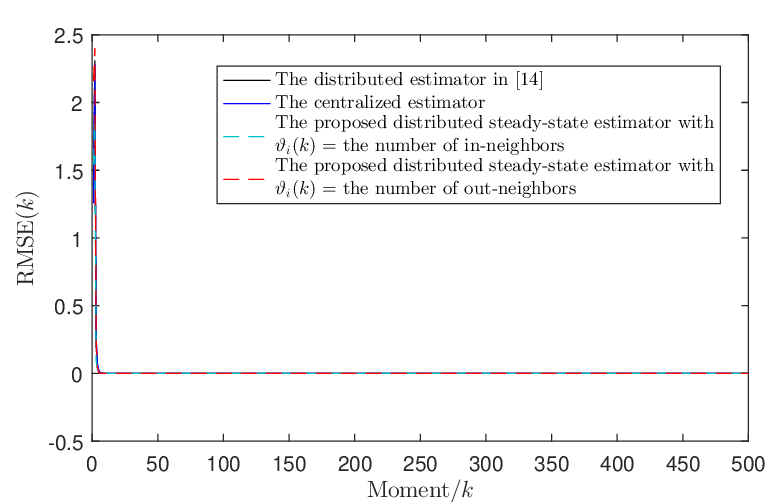}
      \caption{The RMSE of different methods in the time-invariant LIS without noises.}
\end{figure}

We then proceed to evaluate the performance of the proposed distributed estimator in the time-invariant case. Given that the proposed estimator converges in this context, we perform the steady-state version to reduce the real-time communication burden among subsystems. Figs. 3 and 4 show the RMSE of the proposed distributed estimator with and without system noises, respectively. Clearly, the proposed methods still achieve good estimation performance in the time-invariant case. Moreover, the simulation illustrates that the performance of the proposed distributed estimator is minimally affected by the decoupling variable $\vartheta_i$. In fact, it can be verified that $\sum^s_{i=1}\mathbb{I}_i=\sum^s_{i=1}\mathbb{O}_i=|\{(i,j):A_{ij}\neq0,i\neq j\}|$, which indirectly shows that the two different options, $\vartheta_i=\mathbb{I}_i$ or $\vartheta_i=\mathbb{O}_i$, has similar effects on the overall performance of the proposed distributed estimator.

\section{Conclusion}
In this paper, a fully distributed estimator was presented for the LIS, whose local gain was obtained by recursively solving the DMRE with decoupling variables. We developed Lyapunov-based and Markov-based methods for stability analysis to demonstrate that the stability of each subsystem is independent of the global LIS when the decoupling variable coincides either with the number of out-neighbors or in-neighbors. We also showed that any LIS can be equivalently represented as a Markov system. It was proved that the distributed estimator is stable if the DMRE is uniformly bounded from above. Additionally, we established that the local uniform detectability guarantees the stability of the distributed estimator under the condition of weak coupling. When the LIS is time-invariant, we formulated a necessary and sufficient condition for bounding the DMRE using the LMI. Based on the boundedness result, we also demonstrated that the distributed estimator converges to a unique steady state for any initial condition. Finally, the effectiveness of the proposed method was verified using a multi-area power system.

\section*{Appendix}
The appendix shows some lemmas used in this paper. Although these lemmas are trivial generalizations of existing conclusions, we give their proofs for self-contained purposes.

\setcounter{lemma}{0}
\renewcommand{\thelemma}{A\arabic{lemma}}

\begin{lemma}\label{LemmaA1}
Let $\{X(k)\}_{k\geq0}$ and $\{Y(k)\}_{k\geq1}$ be two bounded matrix sequences. The time-varying matrix pair $\big(Y(k)X(k-1),X(k-1)\big)$ is uniformly detectable if the time-varying matrix pair $\big(Y(k),X(k)\big)$ is uniformly detectable.
\end{lemma}

\begin{proof}
Consider the closed-loop system
\begin{equation}\begin{aligned}
\bar{\zeta}(k)=\bar{X}(k-1)\bar{\zeta}(k-1),
\end{aligned}\end{equation}
where
\begin{equation}\begin{aligned}
\bar{X}(k-1)=X(k-1)+Z(k)Y(k)X(k-1).\nonumber
\end{aligned}\end{equation}
Then, one can obtain
\begin{equation}\begin{aligned}
\bar{\zeta}(k+1)=&\big(I+Z(k)Y(k)\big)\hat{X}(k-1)\\
&\times\cdots\times\hat{X}(1)X(0)\bar{\zeta}(0),
\end{aligned}\end{equation}
where
\begin{equation}\begin{aligned}
\hat{X}(k)=X(k)+X(k)Z(k)Y(k)\big).\nonumber
\end{aligned}\end{equation}
According to (85) and the Kalman filtering theory \cite{Optimalfiltering}, there exists a bounded matrix sequence $\{Z(k)\}_{k\geq1}$ such that the closed-loop system (84) is exponentially stable if $\big(Y(k),X(k)\big)$ is uniformly detectable. Then, it follows from the corollary 3.4 in \cite{Anderson0319002} that $\big(Y(k)X(k-1),X(k-1)\big)$ is uniformly detectable. The proof is completed.
\end{proof}

\begin{lemma}\label{LemmaA2}
$\|X_N\times\cdots\times X_1\|\leq\alpha_1$ implies that $\|(X_N\otimes X_N)\times\cdots\times (X_1\otimes X_1)\|\leq \sqrt{n}\alpha_1^2$, where $X_1,\cdots,X_N\in\mathbb{R}^{n\times n}$.
\end{lemma}

\begin{proof}
It is obvious that
\begin{equation}\begin{aligned}
\|X_N\times\cdots\times X_1Y(X_N\times\cdots\times X_1)^\top\|\leq\alpha_1^2
\end{aligned}\end{equation}
for any $Y$ satisfying $\|Y\|=1$. Then, utilizing the vectorization of matrices yields
\begin{equation}\begin{aligned}
&\mathrm{vec}[X_N\times\cdots\times X_1Y(X_N\times\cdots\times X_1)^\top]\\
&\ \ \ \ \ \ \ \ \ \ =\big((X_N\otimes X_N)\times\cdots\times (X_1\otimes X_1)\big)\mathrm{vec}[Y].
\end{aligned}\end{equation}
Recall the definition of the matrix norm, one knows that there exists $Y$ such that
\begin{equation}\begin{aligned}
&\|\big((X_N\otimes X_N)\times\cdots\times (X_1\otimes X_1)\big)\mathrm{vec}[Y]\|\\
&\ \ \ \ \ \ \ \ \ \ =\|\big((X_N\otimes X_N)\times\cdots\times (X_1\otimes X_1)\big)\|.
\end{aligned}\end{equation}
Moreover, it is easy to verify $\|\mathrm{vec}[X]\|\leq \sqrt{n}\|X\|$ for any $X\in\mathbb{R}^{n\times n}$. With this inequality, the lemma can be obtained from (86)-(88). The proof is completed.
\end{proof}

\begin{lemma}\label{LemmaA3}
Let $\mathfrak{T}(X)=\sum^s_{i=1}A_iXA_i^\top$, then the adjoint operator of $\mathfrak{T}$ is given by $\mathfrak{T}^*(X)=\sum^s_{i=1}A_i^\top XA_i$.
\end{lemma}

\begin{proof}
It is trivial that
\begin{equation}\begin{aligned}
&\langle\mathfrak{T}(X),Y\rangle=\mathrm{Tr}(Y^\top\mathfrak{T}(X))=\mathrm{Tr}\big((\mathfrak{T}(X))^\top Y\big)\\
=&\sum^s_{i=1}\mathrm{Tr}(A_iX^\top A_i^\top Y)=\sum^s_{i=1}\mathrm{Tr}(X^\top A_i^\top YA_i)\\
=&\mathrm{Tr}(X^\top\sum^s_{i=1}(A_i^\top YA_i))=\mathrm{Tr}(\sum^s_{i=1}(A_i^\top YA_i)X^\top)\\
=&\langle X,\mathfrak{T}^*(Y)\rangle.\nonumber
\end{aligned}\end{equation}
The proof is completed.
\end{proof}


\section*{Reference}
\begin{thebibliography}{10}

\bibitem{Feddema1068004}
J.~Feddema, C.~Lewis, and D.~Schoenwald, ``Decentralized control of cooperative
  robotic vehicles: theory and application,'' {\em IEEE Transactions on
  Robotics and Automation}, vol.~18, no.~5, pp.~852--864, 2002.

\bibitem{riverso2012hycon2}
S.~Riverso and G.~Ferrari-Trecate, ``Hycon2 benchmark: Power network system,''
  {\em arXiv preprint arXiv:1207.2000}, 2012.

\bibitem{FARINA201838}
M.~Farina, X.~Zhang, and R.~Scattolini, ``A hierarchical multi-rate mpc scheme
  for interconnected systems,'' {\em Automatica}, vol.~90, pp.~38--46, 2018.

\bibitem{Boem8305620}
F.~Boem, S.~Riverso, G.~Ferrari-Trecate, and T.~Parisini, ``Plug-and-play fault
  detection and isolation for large-scale nonlinear systems with stochastic
  uncertainties,'' {\em IEEE Transactions on Automatic Control}, vol.~64,
  no.~1, pp.~4--19, 2019.

\bibitem{Khan4547458}
U.~A. Khan and J.~M.~F. Moura, ``Distributing the kalman filter for large-scale
  systems,'' {\em IEEE Transactions on Signal Processing}, vol.~56, no.~10,
  pp.~4919--4935, 2008.

\bibitem{Haber6553105}
A.~Haber and M.~Verhaegen, ``Moving horizon estimation for large-scale
  interconnected systems,'' {\em IEEE Transactions on Automatic Control},
  vol.~58, no.~11, pp.~2834--2847, 2013.

\bibitem{CHEN2019228}
B.~Chen, G.~Hu, D.~W. Ho, and L.~Yu, ``Distributed kalman filtering for
  time-varying discrete sequential systems,'' {\em Automatica}, vol.~99,
  pp.~228--236, 2019.

\bibitem{Chen9416784}
B.~Chen, G.~Hu, D.~W. Ho, and L.~Yu, ``Distributed estimation and control for
  discrete time-varying interconnected systems,'' {\em IEEE Transactions on
  Automatic Control}, vol.~67, no.~5, pp.~2192--2207, 2022.

\bibitem{ZHANG2023111144}
Y.~Zhang, B.~Chen, and L.~Yu, ``Distributed zonotopic estimation for
  interconnected systems: A fusing overlapping states strategy,'' {\em
  Automatica}, vol.~155, p.~111144, 2023.

\bibitem{Wang9385997}
Y.~Wang, J.~Xiong, and D.~W.~C. Ho, ``Distributed lmmse estimation for
  large-scale systems based on local information,'' {\em IEEE Transactions on
  Cybernetics}, vol.~52, no.~8, pp.~8528--8536, 2022.

\bibitem{Mu10364032}
Y.~Mu, H.~Su, D.~Yang, and H.~Zhang, ``Distributed state estimation design for
  discrete-time interconnected singular systems: A lmi method,'' in {\em 2023
  International Annual Conference on Complex Systems and Intelligent Science
  (CSIS-IAC)}, pp.~726--731, 2023.

\bibitem{Zhang10480463}
Y.~Zhang, B.~Chen, L.~Yu, and D.~W. Ho, ``Distributed estimation for
  interconnected systems with arbitrary coupling structures,'' {\em IEEE
  Transactions on Network Science and Engineering}, vol.~11, no.~4,
  pp.~3667--3677, 2024.

\bibitem{Stefano1040844}
S.~Riverso, D.~Rubini, and G.~Ferrari-Trecate, ``Distributed bounded-error
  state estimation based on practical robust positive invariance,'' {\em
  International Journal of Control}, vol.~88, no.~11, pp.~2277--2290, 2015.

\bibitem{Riverso6760656}
S.~Riverso, M.~Farina, R.~Scattolini, and G.~Ferrari-Trecate, ``Plug-and-play
  distributed state estimation for linear systems,'' in {\em 52nd IEEE
  Conference on Decision and Control}, pp.~4889--4894, 2013.

\bibitem{FARINA2010910}
M.~Farina, G.~Ferrari-Trecate, and R.~Scattolini, ``Moving-horizon
  partition-based state estimation of large-scale systems,'' {\em Automatica},
  vol.~46, no.~5, pp.~910--918, 2010.

\bibitem{Roshany6075282}
S.~Roshany-Yamchi, M.~Cychowski, R.~R. Negenborn, B.~De~Schutter, K.~Delaney,
  and J.~Connell, ``Kalman filter-based distributed predictive control of
  large-scale multi-rate systems: Application to power networks,'' {\em IEEE
  Transactions on Control Systems Technology}, vol.~21, no.~1, pp.~27--39,
  2013.

\bibitem{UHLMANN2003201}
J.~K. Uhlmann, ``Covariance consistency methods for fault-tolerant distributed
  data fusion,'' {\em Information Fusion}, vol.~4, no.~3, pp.~201--215, 2003.

\bibitem{Farina7762729}
M.~Farina and R.~Carli, ``Partition-based distributed kalman filter with plug
  and play features,'' {\em IEEE Transactions on Control of Network Systems},
  vol.~5, no.~1, pp.~560--570, 2018.

\bibitem{FARINA2023100820}
M.~Farina and M.~Rocca, ``A novel distributed algorithm for estimation and
  control of large-scale systems,'' {\em European Journal of Control}, vol.~72,
  p.~100820, 2023.

\bibitem{Siljak1978LargeScaleDS}
D.~D. Siljak, ``Large-scale dynamic systems: Stability and structure,'' 1978.

\bibitem{Optimalfiltering}
B.~D. Anderson and J.~B. Moore, {\em Optimal filtering}.
\newblock Courier Corporation, 2012.

\bibitem{rugh1996linear}
W.~J. Rugh, {\em Linear system theory}.
\newblock Prentice-Hall, Inc., 1996.

\bibitem{ZHANG201784}
Q.~Zhang, ``On stability of the kalman filter for discrete time output error
  systems,'' {\em Systems \& Control Letters}, vol.~107, pp.~84--91, 2017.

\bibitem{Anderson0319002}
B.~D.~O. Anderson and J.~B. Moore, ``Detectability and stabilizability of
  time-varying discrete-time linear systems,'' {\em SIAM Journal on Control and
  Optimization}, vol.~19, no.~1, pp.~20--32, 1981.

\bibitem{moore1980coping}
J.~Moore and B.~D. Anderson, ``Coping with singular transition matrices in
  estimation and control stability theory,'' {\em International Journal of
  Control}, vol.~31, no.~3, pp.~571--586, 1980.

\bibitem{boyd1994linear}
S.~Boyd, L.~El~Ghaoui, E.~Feron, and V.~Balakrishnan, {\em Linear matrix
  inequalities in system and control theory}.
\newblock SIAM, 1994.

\bibitem{Sinopoli1333199}
B.~Sinopoli, L.~Schenato, M.~Franceschetti, K.~Poolla, M.~Jordan, and
  S.~Sastry, ``Kalman filtering with intermittent observations,'' {\em IEEE
  Transactions on Automatic Control}, vol.~49, no.~9, pp.~1453--1464, 2004.

\bibitem{WANG1992139}
Y.~Wang, L.~Xie, and C.~E. {de Souza}, ``Robust control of a class of uncertain
  nonlinear systems,'' {\em Systems \& Control Letters}, vol.~19, no.~2,
  pp.~139--149, 1992.

\bibitem{vidyasagar2002nonlinear}
M.~Vidyasagar, {\em Nonlinear systems analysis}.
\newblock SIAM, 2002.

\bibitem{Fan992124}
F.~Wang and V.~Balakrishnan, ``Robust kalman filters for linear time-varying
  systems with stochastic parametric uncertainties,'' {\em IEEE Transactions on
  Signal Processing}, vol.~50, no.~4, pp.~803--813, 2002.

\bibitem{costa2006discrete}
O.~L.~V. Costa, M.~D. Fragoso, and R.~P. Marques, {\em Discrete-time Markov
  jump linear systems}.
\newblock Springer Science \& Business Media, 2006.

\bibitem{FARINA20133411}
M.~Farina, P.~Colaneri, and R.~Scattolini, ``Block-wise discretization
  accounting for structural constraints,'' {\em Automatica}, vol.~49, no.~11,
  pp.~3411--3417, 2013.
\end{thebibliography}
\end{document}